\documentclass{llncs}
\usepackage{thmtools,thm-restate}
\usepackage{amsmath,amssymb,amsfonts}
\usepackage{mathtools}

\usepackage[usenames,dvipsnames]{color}
\usepackage{xcolor}
\usepackage{xspace}
\usepackage{microtype}
\usepackage{todonotes}
\usepackage{colonequals}
 \usepackage{booktabs}
\usepackage{cite}
\usepackage{tcolorbox}
\tcbuselibrary{theorems}

\usepackage{multirow}
\usepackage{multicol}

\usepackage{paralist}

\usepackage{subfigure}
\usepackage{url}
\usepackage{appendix}
\usepackage{graphicx}
\usepackage{algorithm}
\usepackage[noend]{algorithmic}
\usepackage{listings}
\usepackage{nicefrac}
\usepackage{tikz} 
\usepackage{pgfplots}
\usepackage{mdframed}
\usepackage[super]{nth}
\usepackage{ dsfont }
\usepackage{pifont}

\usetikzlibrary{arrows,decorations.pathmorphing,positioning,fit,trees,shapes,shadows,automata,calc,decorations.pathmorphing} 

\tikzset{outline/.style args={#1}{%
  draw=#1,thick,fill=#1!50},initial text={}}

\usepackage{etoolbox}

\newtoggle{TR}
\toggletrue{TR}
\title{Are Parametric Markov Chains Monotonic?}
\author{Jip Spel, Sebastian Junges, Joost-Pieter Katoen}
\institute{RWTH Aachen University, Aachen, Germany\thanks{Supported by the DFG RTG 2236 ``UnRAVeL''.}
}
\tikzset{every state/.style={minimum size=2.2em}}

\usepackage{twoopt}

\algsetup{linenosize=\scriptsize}

\newcommand{\LB}{\ensuremath{\textsf{lb}}}
\newcommand{\UB}{\ensuremath{\textsf{ub}}}
\newcommand{\ownsubsection}[1]{\medskip\noindent\textbf{#1}}

\newcommand{\ie}{i.e.\@\xspace}

\newcommand{\prophesy}{\textrm{PROPhESY}\xspace}
\newcommand{\prism}{\textrm{PRISM}\xspace}
\newcommand{\storm}{\textrm{Storm}\xspace}
\newcommand{\param}{\textrm{PARAM}\xspace}

\newcommand{\dtmc}{\mathcal{M}}

\newcommand{\pDtmcInit}[1][]{\ensuremath{\dtmc{#1}=(S{#1},\sinit{#1},T{#1},\Paramvar{#1},\probdtmc{#1})}}

\newcommand{\pr}{\ensuremath{\mathrm{Pr}}}
\newcommand{\reachPr}[3]{\ensuremath{\pr^{#1}_{#2}(\finally #3)}}
\newcommandtwoopt{\reachPrT}[2][][]{\ensuremath{\reachPr{#1}{#2}{T}}}

\newcommand{\finally}{\lozenge}

\newcommand{\Succ}{\ensuremath{\textsf{succ}}}

\newcommand{\colorpar}[3]{\colorbox{#1}{\parbox{#2}{#3}}}
\newcommand{\marginremark}[3]{\marginpar{\colorpar{#2}{\linewidth}{\color{#1}#3}}}

\newcommand{\sj}[1]{\marginremark{white}{blue}{\scriptsize{[SJ]~ #1}}}

\newcommand{\R}{\mathbb{R}}
\newcommand{\Q}{\mathbb{Q}}

\newcommand{\Ireal}{[0,\, 1]\subseteq\mathbb{R}}  

\newcommand{\Distr}{\mathit{Distr}}

\newcommand{\distDom}{X}

\newcommand{\distFunc}{\mu}
\newcommand{\distDomElem}{x}

\newcommand{\Var}{\ensuremath{V}\xspace}        
\newcommand{\Paramvar}{\ensuremath{{V}}\xspace}        

\newcommand{\State}{\ensuremath{\mathbb{S}}\xspace}      

\newcommand{\rograph}{RO-graph}

\newcommand{\sinit}{s_{\mathit{I}}} 
\newcommand{\mdp}{\mathcal{M}}

\newcommand{\probdtmc}{\mathcal{P}}

\newcommand{\pdtmc}{\ensuremath{\mathcal{M}}}

\renewcommand{\Pr}{\ensuremath{\textnormal{Pr}}}

\newcommand{\bad}{\ensuremath{\lightning}}

\newcommand{\Paths}{\mbox{\sl Paths}}
\newcommand{\PathsLength}[1]{\mbox{\sl Paths}^{#1}}

\newcommand{\sol}[3][]{\ensuremath{\mathsf{Pr}_{}^{#1 \rightarrow #2}}}

\newcommand{\solWithM}[3][]{\ensuremath{\mathsf{Pr}_{#3}^{#1 \rightarrow #2}}}
\newcommand{\monIncr}[2][]{\ensuremath{#2{\uparrow_{#1}^{R}}}}
\newcommand{\monIncrNoR}[2][]{\ensuremath{#2{\uparrow_{#1}}}}
\newcommand{\monIncrLocal}[2][]{\ensuremath{#2{\uparrow_{#1}^{\ell, R}}}}
\newcommand{\monIncrLocalSteps}[3][]{\ensuremath{#2{\uparrow_{#1}^{\ell, {#3}, R}}}}
\newcommand{\monDecr}[2][]{\ensuremath{#2{\downarrow_{#1}^R}}}
\newcommand{\monDecrNoR}[2][]{\ensuremath{#2{\downarrow_{#1}}}}
\newcommand{\monDecrLocal}[2][]{\ensuremath{#2{\downarrow_{#1}^{\ell, R}}}}
\newcommand{\monDecrLocalSteps}[3][]{\ensuremath{#2{\downarrow_{#1}^{\ell, {#3}, R}}}}

\newcommand{\derivative}[2]{\ensuremath{\dfrac{\partial}{\partial #1}{#2}}}

\newcommand{\reachOrder}[1][]{\ensuremath{\prec_{#1}}}
\newcommand{\reachOrderAssumption}{\ensuremath{\prec^{\mathcal{A}}}}
\newcommand{\reachOrderEq}[1][]{\ensuremath{\preceq_{#1}}}
\newcommand{\reachOrderEqAssumption}{\ensuremath{\preceq^{\mathcal{A}}}}

\DeclareRobustCommand{\good}{\Simley{0.5}{0.2}\xspace}
\DeclareRobustCommand{\bad}{\Simley{-0.5}{0.2}\xspace}
\newcommand{\Simley}[2]{%
	\begin{tikzpicture}[scale=#2]
	\newcommand*{\SmileyRadius}{1.0}%
	;
	
	\pgfmathsetmacro{\eyeX}{0.5*\SmileyRadius*cos(30)}
	\pgfmathsetmacro{\eyeY}{0.5*\SmileyRadius*sin(30)}
	\draw [line width=0.25mm] (\eyeX-0.25,\eyeY) -- (\eyeX-0.25,\eyeY+0.375);
	\draw [line width=0.25mm] (-\eyeX+0.25,\eyeY) -- (-\eyeX+0.25,\eyeY+0.375);
	
	\pgfmathsetmacro{\xScale}{2*\eyeX/180}
	\pgfmathsetmacro{\yScale}{1.0*\eyeY}
	\draw[line width=0.25mm, domain=-\eyeX:\eyeX]
	plot ({\x},{
		-0.1+#1*0.15 
		-#1*1.75*\yScale*(sin((\x+\eyeX)/\xScale))-\eyeY});
	\end{tikzpicture}%
}%

\DeclareMathAlphabet{\mathpzc}{OT1}{pzc}{m}{it}
\def\presuper#1#2%
  {\mathop{}%
   \mathopen{\vphantom{#2}}^{#1}%
   \kern-\scriptspace%
   #2}

\spnewtheorem{definition}{Definition}{\bfseries}{\itshape}

\spnewtheorem{defboxed}[definition]{Definition}{\bfseries}{\itshape}

\newtcolorbox{mymathbox}[1][]{colback=white, sharp corners, #1}
\tcolorboxenvironment{defboxed}{colframe=black, colback=white, sharp corners}

\newcommand{\mathtext}[1]{\text{#1}}

\begin{document}
\lstset{mathescape=true, tabsize=2}

\maketitle
\begin{abstract}
This paper presents a simple algorithm to check whether reachability probabilities in parametric Markov chains are monotonic in (some of) the parameters. The idea is to construct---only using the graph structure of the Markov chain and local transition probabilities---a pre-order on the states. Our algorithm cheaply checks a sufficient condition for monotonicity. Experiments show that monotonicity in several benchmarks is automatically detected, and monotonicity can speed up parameter synthesis up to orders of magnitude faster than a symbolic baseline.
\end{abstract}

\section{Introduction}
Probabilistic model checking~\cite{DBLP:conf/lics/Katoen16,DBLP:reference/mc/BaierAFK18}  takes as input a Markov model together with a specification typically given in a probabilistic extension of LTL or CTL.
The key problem is computing the reachability probability to reach a set of target states.
Efficient probabilistic model checkers include \prism~\cite{DBLP:conf/cav/KwiatkowskaNP11} and \storm~\cite{DBLP:conf/cav/DehnertJK017}.
A major practical obstacle is that transition probabilities need to be precisely given.
Uncertainty about such quantities can be treated by specifying transition probabilities by intervals, as in interval Markov chains~\cite{DBLP:conf/lics/JonssonL91,DBLP:conf/fossacs/ChatterjeeSH08}, or by parametric Markov chains~\cite{Daws04}, which allow for expressing complex parameter dependencies. 

This paper considers parametric Markov chains (pMCs). 
Their transition probabilities are given by arithmetic expressions over real-valued parameters. 
A pMC represents an uncountably large family of Markov chains (MCs): each parameter value from the parameter space induces an MC. 
Reachability properties are easily lifted to pMCs; they are satisfied for a subset of the family of MCs, or equivalently, for a subset of the parameter values. 
Key problems are e.g., is there a parameter valuation such that a given specification $\varphi$ is satisfied (feasibility)?, do all parameter values within a given parameter region satisfy $\varphi$ (verification)?, which parameter values do satisfy $\varphi$ (synthesis)?, and for which parameter values is the probability of satisfying $\varphi$ maximal (optimal synthesis)?
Applications of pMCs include model repair~\cite{DBLP:conf/tacas/BartocciGKRS11,DBLP:conf/tase/ChenHHKQ013,DBLP:conf/nfm/PathakAJTK15,DBLP:journals/iandc/Chatzieleftheriou18,GOUBERMAN201932}, strategy synthesis in AI models such as partially observable MDPs~\cite{DBLP:conf/uai/Junges0WQWK018}, and optimising randomised distributed algorithms~\cite{DBLP:conf/srds/AflakiVBKS17}. 
\prism and \storm, as well as dedicated tools including \param~\cite{param_sttt} and \prophesy~\cite{DBLP:conf/cav/DehnertJJCVBKA15} support pMC analysis. 

Despite the significant progress in the last years in analysing pMCs~\cite{DBLP:conf/atva/QuatmannD0JK16,DBLP:journals/acta/CeskaDPKB17,DBLP:conf/atva/CubuktepeJJKT18}, the scalability of algorithms severely lacks behind methods for ordinary MCs.
There is little hope to overcome this gap. 
The feasibility problem for a reachability probability exceeding $1/2$ is ETR-complete (thus NP-hard)~\cite{DBLP:journals/corr/abs-1904-01503}.
Experiments show that symbolic computations rather than (floating-point) numeric computations have a major impact on analysis times~\cite{DBLP:conf/atva/QuatmannD0JK16}.

This paper takes a different approach and focuses on \emph{monotonicity}, in particular on (a) an algorithm to check whether  pMCs are monotonic in (some of) the parameters with respect to reachability probabilities, and (b) on investigating to what extent monotonicity can be exploited to accelerate parameter synthesis.
Monotonicity has an enormous potential to simplify pMC analysis; e.g., checking whether all points within a rectangle satisfy $\varphi$ reduces to checking whether a line fragment satisfies $\varphi$ when one parameter is monotonic.
Thus, the verification problem for an $n{+}k$-dimensional hyper-rectangle reduces to checking an $n$-dimensional rectangle when the pMC at hand is monotonic in $k$ parameters.
Similarly, determining a parameter instantiation that maximises the probability of $\varphi$ (optimal synthesis) simplifies considerably if all---just a single instance suffices---or some parameters are monotone. 
Similar problems at the heart of model repair~\cite{DBLP:conf/tacas/BartocciGKRS11,DBLP:conf/tase/ChenHHKQ013,DBLP:conf/nfm/PathakAJTK15,DBLP:journals/iandc/Chatzieleftheriou18,GOUBERMAN201932} also substantially benefit from monotonicity.

Unfortunately, determining monotonicity is as hard as parameter synthesis.
The key idea therefore is to construct---using the graph structure of the pMC and local transition probabilities---a pre-order on the states that is used to check a sufficient condition for monotonicity. 
The paper gradually develops a semi-decision algorithm, starting with acyclic pMCs, to the general setting with cycles.
The algorithm uses assumptions indicating whether a state is below (or equivalent to) another one, and techniques are described to discharge these assumptions.
Possible outcomes of our algorithms are: a pMC is monotonic increasing in a certain parameter for a given region, monotone decreasing, or unknown.
Experiments with a prototypical implementation built on top of \storm show that monotonicity is detected automatically and scalable in several benchmarks from the literature. 
In addition, exploiting monotonicity in a state-of-the-art parameter synthesis can lead to speed-ups of up to an order of magnitude.
(Proofs of our results can be found in the appendix.)

\section{Preliminaries and Problem Statement}
\label{sec:preliminaries}

A \emph{probability distribution} over a finite or countably infinite set $\distDom$ is a function $\distFunc\colon\distDom\rightarrow\Ireal$ with $\sum_{\distDomElem\in\distDom}\distFunc(\distDomElem)=1$. 
The set of all distributions on $\distDom$ is denoted by $\Distr(\distDom)$.
Let $\vec{a} \in \R^n$ denote $(a_1, \hdots, a_n)$, and $\vec{e}_i$ denote the vector with $e_j = 1$ if $i=j$ and $e_j = 0$ otherwise.
The set of multivariate polynomials over ordered variables $\vec{x} = (x_1,\hdots,x_n)$ is denoted $\mathbb{Q}[\vec{x}]$. 
An \emph{instantiation} for a finite set $\Paramvar$ of real-valued variables is a function $u\colon V \rightarrow \R$. 
We typically denote $u$ as a vector $\vec{u} \in \R^n$ with $u_i \colonequals u(x_i)$.
A polynomial $f$ can be interpreted as a function $f\colon \R^n \rightarrow \R$, where $f(\vec{u})$ is obtained by substitution i.e., $f[\vec{x} \leftarrow \vec{u}]$, where each occurrence of $x_i$ in $f$ is replaced by $u(x_i)$.

\begin{definition}[Multivariate monotonic function]
	\label{def:monotoneFunction}
	A function $f\colon \R^n \rightarrow \R$ is \emph{monotonic increasing in $x_i$ on set $R \subset \R^n$}, denoted $\monIncr[x_i]{f}$, if \[ f(\vec{a}) \leq f(\vec{a} + b \cdot \vec{e}_i) \qquad \forall \vec{a} \in R \; \forall b \in \R_{\geq 0}.\]
	A function $f$ is \emph{monotone decreasing in $x_i$ on $R$}, denoted $\monDecr[x_i]{f}$, if $\monIncr[x_i]{({-}f)}$.
	A function $f$ is \emph{monotone increasing (decreasing) on $R$}, denoted $\monIncr{f}$ ($\monDecr{f}$), if $\monIncr[x_i]{f}$ ($\monDecr[x_i]{f}$) for all $x_i \in \vec{x}$, respectively.
\end{definition}
If function $f$ is continuously differentiable on the open set $R \subset \R^n$, then \linebreak\({\forall \vec{u}\in R}.~\derivative{x_i}{f(\vec{u})} \geq 0 \implies \monIncr[x_i]{f}. \)
In particular, any $f \in \Q[\vec{x}]$ is continuously differentiable on $\R^n$.

\begin{definition}[pMC]\label{def:pmdp}
A \emph{parametric Markov Chain (pMC)} is a tuple $\pDtmcInit$ with a finite set $S$ of \emph{states}, an \emph{initial state} $\sinit \in S$, a finite set $T \subseteq S$ of \emph{target states}, a finite set $\Paramvar$ of real-valued variables \emph{(parameters)} and a \emph{transition function} $\probdtmc \colon S \times S \rightarrow \mathbb{Q}[V]$.

\end{definition}

We define $\Succ(s) = \{ s' \in S \mid \probdtmc(s,s') \neq 0 \}$.
A pMC $\pdtmc$ is a \emph{(discrete-time) Markov chain} (MC) if the transition function yields \emph{well-defined} probability distributions, \ie, $\probdtmc(s, \cdot) \in \Distr(S)$ for each $s\in S$. 
A state $s$ is called \emph{parametric}, if $\probdtmc(s,s') \not\in \Q$ for some $s' \in S$. 
Applying an \emph{instantiation} $\vec{u}$ to a pMC $\pdtmc$ yields $\pdtmc[\vec{u}]$ by replacing each $f\in\mathbb{Q}[V]$ in $\dtmc$ by $f(\vec{u})$.
An instantiation $\vec{u}$ is \emph{well-defined} (for $\mdp$) if the $\mdp[\vec{u}]$ is an MC.
A well-defined instantiation $\vec{u}$ is \emph{graph-preserving} (for $\mdp$) if the topology is preserved, that is, for all $s,s' \in S$ with $\probdtmc(s,s') \neq 0$ implies $\probdtmc(s,s')(\vec{u}) \neq 0$.
A set of instantiations is called a \emph{region}.
A region $R$ is well-defined (graph-preserving) if  $\forall \vec{u} \in R$, $\vec{u}$ is well-defined (graph-preserving).

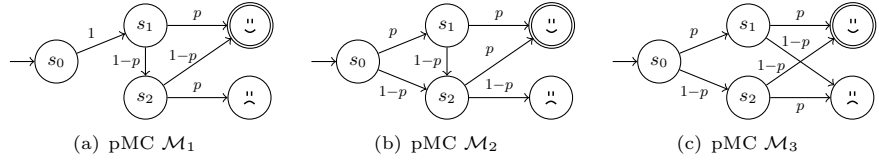
\begin{figure}[t]
\centering
\subfigure[pMC $\pdtmc_1$]{
\begin{tikzpicture}[scale=0.8, every node/.style={scale=0.8}, node distance=0.8cm]
		\node[state, initial]           (s0)           {$s_0$};
		\node[state]           (s1) [right=0.6cm of s0, yshift=0.6cm]               {$s_1$};
		\node[state]           (s2) [right=0.6cm of s0, yshift=-0.6cm]               {$s_2$};
		\node[state, accepting]           (g) [right= of s1]               {$\good$};
		\node[state]           (b) [right= of s2]               {$\bad$};
		
		\draw[->] (s0) edge node[auto, scale=0.8] {$1$} (s1);
		\draw[->] (s1) edge node[left, scale=0.8] {$1{-}p$} (s2);
		\draw[->] (s1) edge node[auto, scale=0.8] {$p$} (g);
		\draw[->] (s2) edge node[left, scale=0.8, xshift=0.1cm, yshift=0.1cm] {$1{-}p$} (g);
		\draw[->] (s2) edge node[auto, scale=0.8] {$p$} (b);
		
\end{tikzpicture}
\label{fig:nonmon}
}
\subfigure[pMC $\pdtmc_2$]{
\begin{tikzpicture}[scale=0.8, every node/.style={scale=0.8}, node distance=0.8cm]
		\node[state, initial]           (s0)           {$s_0$};
		\node[state]           (s1) [right=0.6cm of s0, yshift=0.6cm]               {$s_1$};
		\node[state]           (s2) [right=0.6cm of s0, yshift=-0.6cm]               {$s_2$};
		\node[state, accepting]           (g) [right= of s1]               {$\good$};
		\node[state]           (b) [right= of s2]               {$\bad$};
		
		\draw[->] (s0) edge node[auto, scale=0.8] {$p$} (s1);
		\draw[->] (s0) edge node[pos=0.3,below, scale=0.8, yshift=-.1cm] {$1{-}p$} (s2);
		\draw[->] (s1) edge node[left, scale=0.8] {$1{-}p$} (s2);
		\draw[->] (s1) edge node[auto, scale=0.8] {$p$} (g);
		\draw[->] (s2) edge node[auto, scale=0.8] {$p$} (g);
		\draw[->] (s2) edge node[auto, scale=0.8] {$1{-}p$} (b);
		
\end{tikzpicture}
\label{fig:mon}
}
\subfigure[pMC $\pdtmc_3$]{
\begin{tikzpicture}[scale=0.8, every node/.style={scale=0.8}, node distance=0.8cm]
		\node[state, initial]           (s0)           {$s_0$};
		\node[state]           (s1) [right=0.6cm of s0, yshift=0.6cm]               {$s_1$};
		\node[state]           (s2) [right=0.6cm of s0, yshift=-0.6cm]               {$s_2$};
		\node[state, accepting]           (g) [right= of s1]               {$\good$};
		\node[state]           (b) [right= of s2]               {$\bad$};
		
		\draw[->] (s0) edge node[auto, scale=0.8] {$p$} (s1);
		\draw[->] (s0) edge node[pos=0.3,below, scale=0.8, yshift=-.1cm] {$1{-}p$} (s2);
		\draw[->] (s1) edge node[pos=0.1,xshift=0.05cm,right, scale=0.8] {$1{-}p$} (b);
		\draw[->] (s1) edge node[auto, scale=0.8] {$p$} (g);
		\draw[->] (s2) edge node[pos=0.4, xshift=-0.05cm,left, scale=0.8] {$1{-}p$} (g);
		\draw[->] (s2) edge node[below, scale=0.8] {$p$} (b);
		
\end{tikzpicture}
\label{fig:inc}
}
\vspace{-0.5em}
\caption{Three simple pMCs}
\label{fig:3pmcs}
\end{figure}

\begin{example}
	Fig.~\ref{fig:3pmcs} shows three pMCs, all with a single parameter $p$.
	Instantiation $\vec{u} = \{ p \mapsto 0.4 \}$ 
	is graph-preserving for all these pMCs.
	Instantiation $\vec{u}' = \{ p \mapsto 1 \}$ is well-defined, but not graph-preserving, while $\vec{u}'' = \{ p \mapsto 2 \}$ is not well-defined.	
\end{example}
\begin{remark}
Most pMCs in the literature are linear, i.e., all transition probabilities are linear.  
Many pMCs---including those in Fig.~\ref{fig:3pmcs}---are \emph{simple}, i.e., $\probdtmc(s,s') \in \{ p, 1{-}p \mid p \in \Var \} \cup \Q$ for all $s,s'\in S$.
For simple pMCs, all well-defined instantiations (graph-preserving) are in $[0,1]^{|\Var |}$ (in $(0,1)^{| \Var |}$).
\end{remark}

For a parameter-free MC $\dtmc$, $\reachPrT[s][\dtmc] \in [0,1] \subseteq \R$ denotes the probability that from state $s$ the target $T$ is reached. 
For a formal definition, we refer to, e.g., \cite[Ch.~10]{BK08}.
For pMC $\pdtmc$, $\reachPrT[s][\pdtmc]$ is not a constant, but rather a function $\solWithM[s]{T}{\pdtmc} \colon \Paramvar \rightarrow [0,1]$,
s.t.\ $\solWithM[s]{T}{\pdtmc}(\vec{u}) = \reachPrT[s][{\pdtmc[\vec{u}]}]$.
We call $\solWithM[s]{T}{\pdtmc}$ the \emph{solution function}, and for conciseness, we typically omit $\pdtmc$.
For two graph-preserving instantiations $\vec{u}, \vec{u}'$, we have that $\sol[s]{T}{\pdtmc}(\vec{u}) = 0$ implies $\sol[s]{T}{\pdtmc}(\vec{u}') = 0$ (analogous for ${=}1$).
We simply write $\sol[s]{T}{\pdtmc} = 0$ (or ${=}1$).
\begin{example}
	For the pMC in Fig.~\ref{fig:nonmon}, the solution function $\sol[s]{T}{\pdtmc}$ is $p + (1-p)^2$.
	For the pMCs in Fig.~\ref{fig:mon} and \ref{fig:inc}, it is $-p^3+p^2+p$ and $p^2+(1-p)^2$, respectively.
	\end{example}
The closed-form of $\sol[s]{T}{\pdtmc}$ on a graph-preserving region is a rational function over $V$, i.e., a fraction of two polynomials over $V$. 
Various methods for computing this closed form on a graph-preserving region have been proposed~\cite{Daws04,param_sttt,DBLP:journals/tse/FilieriTG16,DBLP:journals/corr/abs-1709-02093,DBLP:conf/cav/DehnertJJCVBKA15}. 
Such a closed-form can be exponential in the number of parameters~\cite{DBLP:journals/corr/abs-1709-02093}, and is typically (very) large already with one or two parameters~\cite{param_sttt,DBLP:conf/cav/DehnertJJCVBKA15}.
On a graph-preserving region, $\sol[s]{T}{\pdtmc}$ is continuously differentiable~\cite{DBLP:conf/atva/QuatmannD0JK16}.

The parameter feasibility problem considered in e.g.\ \cite{param_sttt,DBLP:conf/cav/DehnertJJCVBKA15,DBLP:conf/atva/QuatmannD0JK16,DBLP:journals/corr/Chonev17,DBLP:journals/corr/abs-1709-02093,DBLP:conf/atva/CubuktepeJJKT18,DBLP:conf/atva/GainerHS18} is: 
\emph{Given a pMC $\pdtmc$, a threshold $\lambda \in [0,1]$, and a graph-preserving region $R$, is there an instantiation $\vec{u} \in R$ s.t.\ $\solWithM[s_I]{T}{\pdtmc}(\vec{u}) \geq \lambda$?}
This problem is square-root-sum hard~\cite{DBLP:journals/corr/Chonev17}. For any fixed number of parameters, this problem is decidable in P~\cite{DBLP:journals/corr/abs-1709-02093}.
\begin{example}
For the pMC in Fig~\ref{fig:nonmon}, $R = [0.4, 0.6]$, and $\lambda = 0.9$, the result to the parameter feasibility is \texttt{false}, as $\max_{\vec{u}\in R} \sol[s]{T}{\pdtmc}(\vec{u}) < 0.9$.
\end{example}
\begin{definition}[Monotonicity in pMCs]
For pMC $\pDtmcInit$, parameter $p \in V$, and graph-preserving region $R$, we call $\pdtmc$ \emph{monotonic increasing} in $p$ on $R$, written $\monIncr[p]{\pdtmc}$, if $\monIncr[p]{\sol[\sinit]{T}{\pdtmc}}$.
\emph{Monotonic decreasing}, written $\monDecr[p]{\pdtmc}$, is defined analogously.	
\end{definition}

\begin{example}
	The pMC in Fig.~\ref{fig:mon} is monotonic in $p$ on $(0,1)$, as its derivative $-3p^2 + 2p + 1$ is strictly positive on $(0,1)$.
	The pMC in Fig.~\ref{fig:nonmon} is not, as witnessed by the derivative $1 - 2(1{-}p)$.
\end{example}
The above example immediately suggests a complete algorithm to decide whether $\monIncr[p]{\pdtmc}$ (or analogously $\monDecr[p]{\pdtmc}$):
Compute the solution function, symbolically compute the derivative w.r.t.\ parameter $p$, and ask a solver (e.g., an SMT-solver for non-linear real arithmetic~\cite{DBLP:journals/cca/JovanovicM12}) for the existence of a negative instantiation in $R$. 
If no such instantiation exists, then $\monIncr[p]{\pdtmc}$.
Observe that the size of the solution function and its derivative are in the same order of magnitude.
This algorithm runs in polynomial time for any fixed number of parameters, yet the practical runtime even for medium-sized problems is unsatisfactory, due to the high costs of the symbolic operations involved.
The result below motivates to look for \emph{sufficient criteria for monotonicity that can be practically efficiently checked}.
\begin{restatable}{theorem}{montheorem}
\label{lem:complexitymonproblem}
pMC verification\footnote{The complement of the parameter feasibility problem.} is polynomial-time reducible to the decision problem whether a pMC is monotonic.
\end{restatable}
\noindent 
Proving \emph{non-}monotonicity is often simpler---finding three instantiations along a line that disprove monotonicity sufficess---, and less beneficial for parameter synthesis. 
This paper focuses on proving monotonicity rather than disproving it.
\begin{example}
The three instantiations on Fig.~\ref{fig:nonmon}: $p \mapsto 0.3, 0.5, 0.9$ yield reachability probabilities: $0.79$, $0.75$, $0.91$. Thus neither $\monIncr[p]{\pdtmc}$ nor $\monDecr[p]{\pdtmc}$ on $R = [0.3,0.9]$.
\end{example}

\subsubsection*{Problem statement.}
Given a pMC $\pdtmc$, a parameter $p$, and a region $R$, construct an \emph{efficient} algorithm that determines either $\monIncr[p]{\pdtmc}$, $\monDecr[p]{\pdtmc}$,  or ``unknown''.
\medskip\par	

In the following, let $\pDtmcInit$ be a pMC with $R$ a graph-preserving region.
Let $\good$ ($\bad$) denote all states $s\in S$ with $\sol[s]{T}{\pdtmc} = 1$ ($\sol[s]{T}{\pdtmc} = 0$).
By a standard preprocessing~\cite{BK08}, we assume a single $\good$ and $\bad$ state.
We call a parameter $p$ monotonic, if the solution function of the pMC is monotonic in $p$.

\section{A Sufficient Criterion for Monotonicity}
\label{sec:ro}

In this section, we combine reasoning about the underlying graph structure of a pMC 
and local reasoning about transition probabilities of single states 
to deduce a sufficient criterion for monotonicity.

\medskip\noindent\textbf{Reachability orders.}
\label{subsec:reachOrder}
\begin{definition}[Reachability order/RO-graph]
\label{def:exhaustiveReachOrder}
An ordering relation \(\reachOrderEq[R,T]\, \subseteq S \times S\) is a \emph{reachability order w.r.t.\ $T \subseteq S$ and region $R$} if for all $s,t \in S$: 
\[ 
s\reachOrderEq[R,T] t 
\quad \text{implies} \quad
\forall \vec{u}\in R.~ \sol[s]{T}{\pdtmc}(\vec{u}) \leq \sol[t]{T}{\pdtmc}(\vec{u}).
\]
The order $\reachOrderEq[R,T]$ is called \emph{exhaustive} if the reverse implication holds too.
The Hasse-diagram\footnote{That is, 
$G = (S,E)$ with $E = \{(s,t) \mid s,t \in S \wedge s \reachOrderEq t \wedge (\not\exists s'\in S.~s\reachOrderEq s' \reachOrderEq t)\} $.} 
for a reachability order is called an \emph{RO-graph}.
\end{definition}

The relation $\reachOrderEq[R,T]$ is a reflexive (aka: non-strict) pre-order.
The exhaustive reachability order is the union of all reachability orders, and always exists.
Let $\equiv_{R,T}$ denote the kernel of $\reachOrderEq[R,T]$, i.e., $\equiv_{R,T} \, = \,  \reachOrderEq[R,T] \, \cap \, \reachOrderEq[R,T]^{{-}1}$.
If $\reachOrderEq[R,T]$ is exhaustive:
\[ s \equiv_{R,T} t \quad \text{ iff } \quad \forall \vec{u} \in R.~ \sol[s]{T}{\pdtmc}(\vec{u}) = \sol[t]{T}{\pdtmc}(\vec{u}).\] 
We often omit the subscript $R,T$ from $\reachOrderEq$ and $\equiv$ for brevity.
Let $[s]$ denote the equivalence class w.r.t.\ $\equiv$, i.e., $[s] = \{ t\in S \mid s \equiv t\}$, and $[S]$ denote the set of equivalence classes on $S$.
We lift $\reachOrderEq$ to sets in a point-wise manner, i.e., $s \reachOrderEq X$ denotes $s \reachOrderEq x$ for all $x \in X$.
In the following, we use w.l.o.g.\ that each reachability order $\reachOrderEq$ satisfies $\bad \reachOrderEq S \setminus \{\bad \}$ and $S \setminus \{ \good \} \reachOrderEq \good$.

\begin{figure}[t]
	\centering
	\subfigure[RO-graph for $\pdtmc_1$]{
	\begin{tikzpicture}[scale=0.8, every node/.style={scale=0.8}, node distance=0.5cm, baseline]
	\node (good) {\good};
	\node (s1)  [left=of good] {$s_2$}; 
	\node (s0) [left=of s1] {$[s_0]$};
	\node ()  [left=of good, yshift=-0.5cm] {}; 
	\node (bad) [left=of s0] {\bad};
	
	\draw[<-] (good) -- (s1);
	\draw[<-] (s1) -- (s0);
	\draw[<-] (s0) -- (bad);
	\end{tikzpicture}
	\label{fig:rograph_small_pmc_exhaustive}
}
\subfigure[RO-graph for $\pdtmc_1$]{
	\begin{tikzpicture}[scale=0.8, every node/.style={scale=0.8}, node distance=0.5cm,baseline]
	\node (good) {\good};
	\node (s1)  [left=of good, yshift=-0.35cm] {\phantom{$[$}$s_2$\phantom{$]$}}; 
	\node (s0) [left=of good, yshift=0.35cm] {$[s_0]$};
	\node (bad) [left=of s0, yshift=-0.35cm] {\bad};
		\node () [right=of good]{};
	\node() [left=of bad]{};
	
	\draw[<-] (good) -- (s1);
	\draw[<-] (good) -- (s0);	
	\draw[<-] (s1) -- (bad);
	\draw[<-] (s0) -- (bad);
	\end{tikzpicture}
	\label{fig:rograph_small_pmc_non_exhaustive}
}
\subfigure[RO-graph for $\pdtmc_3$]{
	\begin{tikzpicture}[scale=0.8, every node/.style={scale=0.8}, node distance=0.5cm, baseline]
	
	\node (good) {\good};
	\node (s1) [left=of good] {$s_1$};
	\node (s0) [left=of good, yshift=-0.35cm] {$s_0$};
	\node (s2) [left=of good, yshift=0.35cm] {$s_2$};
	\node (bad) [left=of s1] {\bad};
	\node () [right=of good]{};
	\node() [left=of bad]{};
	
	\draw[<-] (good) -- (s1);
	\draw[<-] (good) -- (s0);
	\draw[<-] (good) -- (s2);
	\draw[<-] (s1) -- (bad);
	\draw[<-] (s2) -- (bad);
	\draw[<-] (s0) -- (bad);
	\end{tikzpicture}
	\label{fig:rograph_small_pmc_exhaustive_incomparable}
}
%
\label{fig:pMCsWithRO}
\vspace{-0.5em}
\caption{RO-graphs for some of the pMCs in Fig.~\ref{fig:3pmcs}}
\end{figure}
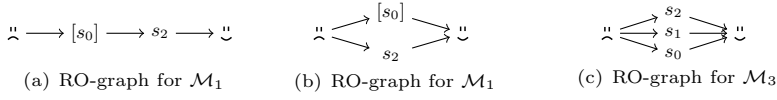 	
\begin{example}
	\label{ex:exhaustiveReachOrder}
	Consider pMC $\pdtmc_1$ in Fig.~\ref{fig:nonmon} with arbitrary region $R$. 
	Fig.~\ref{fig:rograph_small_pmc_exhaustive} shows the {\rograph} of the exhaustive reachability order, with $[s_0] = \{s_0, s_1\}$. 
	Fig.~\ref{fig:rograph_small_pmc_non_exhaustive} shows a non-exhaustive reachability order for $\pdtmc_1$.
	Next, consider Fig.~\ref{fig:inc} with region $R = (0,1)$. 
	States $s_1$ and $s_2$ are incomparable:
	 For $\vec{u}_1\in R$ with $\vec{u}_1(p)< \frac{1}{2}:~ \sol[s_1]{T}{\pdtmc}(\vec{u}_1) < \sol[s_2]{T}{\pdtmc}(\vec{u}_1)$, while for $ \vec{u}_2\in R$ with $\vec{u}_2(p) > \frac{1}{2}:~ \sol[s_1]{T}{\pdtmc}(\vec{u}_2) > \sol[s_2]{T}{\pdtmc}(\vec{u}_2)$.
	 Analogously, $s_0, s_1$ and $s_0, s_2$ are pairwise incomparable.
	Fig.~\ref{fig:rograph_small_pmc_exhaustive_incomparable} depicts the corresponding exhaustive reachability order.	
\end{example}

\medskip\noindent\textbf{Local monotonicity.}
\label{subsec:monotonicity}
Next, we show how a local notion of monotonicity suffices to infer monotonicity.
\begin{definition}[Locally monotonic increasing]
	\label{def:localMonIncr}
	$\sol[s]{T}{\pdtmc}$ is \emph{locally monotonic increasing} in parameter $p$ (at $s$) on region $R$, denoted $\monIncrLocal[p]{\sol[s]{T}{\pdtmc}}$, if $\forall \vec{u} \in R$:
	\[\left(\sum_{s' \in \Succ(s)} \left(\derivative{p}{\probdtmc(s,s')}\right) \cdot \sol[s']{T}{\pdtmc}  \right)(\vec{u}) \geq 0.\]
\end{definition}
\emph{Locally monotonic decreasing}, denoted $\monDecrLocal[p]{\sol[s]{T}{\pdtmc}}$, is defined analogously ($\leq 0$).
Thus, while global monotonicity considers the derivative of the full solution function, local monotonicity only considers the derivative of the first transition.
\begin{example}
	For state $s_0$ in Fig.~\ref{fig:mon}, we compute: \[ \left(\derivative{p}{p}\right) \cdot \sol[s_1]{T}{\pdtmc} + \left(\derivative{p}{(1-p)}\right) \cdot \sol[s_2]{T}{\pdtmc} = 1 \cdot \big(p+(1-p){\cdot}p\big) - 1 \cdot p = p-p^2.\] 
	By checking for which instantiations this function is non-negative, we obtain that $s_0$ is locally monotonic increasing on any graph-preserving $R$.
	Similar computations show that $s_1$ and $s_2$ are locally monotonic increasing.
	In Fig.~\ref{fig:inc}, $s_1$ is locally monotonic increasing, $s_2$ is locally monotonic decreasing, and $s_0$ is neither locally monotonic increasing nor decreasing.
\end{example}
Observe that non-parametric states are monotonic increasing and decreasing in any parameter.
Reachability orders may induce local monotonicity:

\begin{restatable}{lemma}{lemLocalMonSucc}
	\label{lem:localmonsucc3}
	Let $\Succ(s) =\{s_1, \ldots, s_n\}$, $P(s,s_i) = f_i$ and $\forall j > i. s_j \reachOrderEq s_i $. Then:
	\[\monIncrLocal[p]{\sol[s]{T}{\pdtmc}} \quad\text{ iff }\quad \exists i \in [1, \hdots, n].\Big(\forall j \leq i.\; \monIncr[p]{f_j} \text{ and } \forall j > i.\; \monDecr[p]{f_j} \Big).\]
\end{restatable}

\begin{restatable}{theorem}{thmMonIncr}
	\label{thm:monIncr}
	\[\Big(\forall s \in S.\,\monIncrLocal[p]{\sol[s]{T}{\pdtmc}}\Big) \implies \monIncr[p]{\sol[\sinit]{T}{\pdtmc}}.\]	
\end{restatable}
\begin{example}
	Consider Fig.~\ref{fig:mon}, 
	observe that $s_2 \reachOrderEq s_1$.
	Applying Lem.~\ref{lem:localmonsucc3} to $s_0$ with monotonic increasing $f {=} p$,
	yields that $s_0$ is locally monotonic increasing. 
	All states are locally monotonic increasing, thus $\pdtmc_2$ is (globally) monotonic increasing.
\end{example}

\medskip\noindent\textbf{Sufficient reachability orders.}
Above, we only regard the reachability order locally, in order to deduce (local) monotonicity. 
Thus, to deduce (global) monotonicity from a reachability order, it suffices to compute a subset of the exhaustive reachability order. 
\begin{definition}[Sufficient reachability order]
	\label{def:sufficientRO}
	Reachability order 
	$\reachOrderEq[]$ is \emph{sufficient} for $s \in S$ if for all $s_1, s_2 \in \Succ(s)$:  $\left(s_1 \reachOrderEq[] s_2 \,\vee\,s_2 \reachOrderEq[] s_1 \right)$ holds.
	The reachability order is \emph{sufficient} for $\pdtmc$ if it is sufficient for all parametric states.
\end{definition}
A reachability order $\reachOrderEq$ is thus sufficient for $s$ if $(\Succ(s), \reachOrderEq)$ is a total order.
A sufficient reachability order does not necessarily exist.
\begin{example}
	\label{ex:sufficientReachOrder}
	The reachability order in Fig.~\ref{fig:rograph_small_pmc_exhaustive} is sufficient for all states. 
	The reachability order in Fig.~\ref{fig:rograph_small_pmc_exhaustive_incomparable} is not sufficient for $s_0$.
\end{example}

\begin{corollary}
Given a pMC $\pdtmc$ s.t.\ all states $s \in S$ have $|\Succ(s)|\leq2$, and only monotonic transition functions. 
If reachability order $\reachOrderEq$ is sufficient for $s$, then $\sol[s]{T}{\pdtmc}$ is locally monotonic increasing/decreasing on region $R$ in all parameters.
\end{corollary}
The proof follows immediately from Def.~\ref{def:sufficientRO} and Lem.~\ref{lem:localmonsucc3}.
A similar statement holds for the general case with arbitrarily many successors.
The reachability order $\reachOrderEq$ is called a \emph{witness} for monotonicity of parameter $p$ on $R$ whenever either all states are locally increasing or all are locally decreasing in $p$.
A sufficient $\reachOrderEq$ (for $\pdtmc$) does in general not imply global monotonicity of $\pdtmc$.
\begin{example}
\label{ex:inconclusiveorder}
	While the order shown in Fig.~\ref{fig:rograph_small_pmc_exhaustive} is sufficient for pMC $\pdtmc_1$ \linebreak(Fig.~\ref{fig:nonmon}), $\pdtmc_1$ is not monotonic: state $s_1$ is locally increasing, but state $s_2$ is locally decreasing. 
\end{example}
We call such reachability orders (with the pMC) \emph{inconclusive} for $p$ and $R$.

\section{Automatically Proving Monotonicity}
\label{sec:automatic}
In this section, we discuss how to automatically construct a sufficient reachability order to deduce monotonicity of (some of) the parameters in the given pMC.
The following negative result motivates us to consider a heuristic approach:
\begin{restatable}{lemma}{lemcomplexity}
\label{lem:complexity}
pMC verification is polynomial-time reducible to the decision problem whether two states are ordered by the exhaustive reachability order.
\end{restatable}

Our algorithmic approach is based on RO-graphs.
We first consider how these graphs can be used to determine monotonicity (Sect.~\ref{subsec:monousingro}). 
The main part of this section is devoted to constructing RO-graphs.
We start with a basic idea for obtaining reachability orders for acyclic pMCs (Sect.~\ref{subsec:constructro}).
To get \emph{sufficient} orders, the algorithm is refined by automatically making \emph{assumptions}, such as $s \reachOrderEq s'$ and/or  $s' \reachOrderEq s$ (Sect.~\ref{subsec:makingassumptions}).
We then describe how these assumptions can be discharged (Sect.~\ref{sec:discharge}), and finally extend the algorithm to treat cycles (Sect.~\ref{subsec:cycles}).
 
\ownsubsection{Checking monotonicity using a reachability order.}
\label{subsec:monousingro}
The base is to check whether the RO-graph is a witness for monotonicity.
This is done as follows.
Using the RO-graph, we determine global monotonicity of the pMC by checking each parametric state $s$ for local monotonicity (cf. Thm~\ref{thm:monIncr}). 
To decide whether $s$ is local monotonic, we consider the ordering of its direct successors and the derivatives of the probabilistic transition functions and apply Lem.~\ref{lem:localmonsucc3}.

\subsection{Constructing reachability orders}
\label{subsec:constructro}
Our aim is to construct a  (not necessarily sufficient) reachability order from the graph structure of a pMC.
Let us introduce some standard notions.
For reachability order $\reachOrderEq$ and $X \subseteq S$, $\UB(X) = \{ s \in S \mid X \reachOrderEq s \} $ and $\LB(X) = \{ s \in  S \mid s \reachOrderEq X \}$ denote the upper and lower bounds of $X$.
As $\bad \reachOrderEq S$ and $S \reachOrderEq \good$, these sets are non-empty.
Furthermore, let $\min(X) = \{ x \in X \mid {\not\exists} x' \in X. x' \reachOrderEq x \}$, and $\max(X) = \{ x \in X \mid \not\exists x' \in X. x \reachOrderEq x' \}$.
If $(X,\reachOrderEq)$ is a lattice, then it has a unique minimal upper bound (and maximal lower bound). Then:
\begin{restatable}{lemma}{lemReachrelativesucc}
\label{lem:reachrelativesucc}
	For $s \in S$, either $\Succ(s) \subseteq [s]$ or $\LB(\Succ(s)) \reachOrder s \reachOrder \UB(\Succ(s))$.
\end{restatable}
The first case essentially says that if $\exists s' \in \Succ(s)$ with $\Succ(s) \subseteq [s']$, then also $s \in [s']$.
Lem.~\ref{lem:reachrelativesucc} enables to construct reachability orders:
\begin{example}
Reconsider the pMC $\pdtmc_1$ from Fig.~\ref{fig:nonmon}.
Clearly $\bad \reachOrder \good$.
Now consider the pMC in reverse topological order (from back to front).
We start with state $s_2$. 
By Lem.~\ref{lem:reachrelativesucc}, we conclude $\bad \reachOrder s_2 \reachOrder \good$. 
Next, we consider $s_1$, and analogously conclude $s_2 \reachOrder s_1 \reachOrder \good$. 
Finally, considering $s_0$ gives $\Succ(s_0) \subseteq [s_1]$ thus, $s_0 \in [s_1]$. 
The resulting (exhaustive) reachability order is given in Fig.~\ref{fig:rograph_small_pmc_exhaustive}.
\end{example}

This reasoning is automated by algorithm Alg.~\ref{alg:latticeConstruction}.
It takes as input an acyclic pMC and iteratively computes a set of reachability orders, starting from the trivial order $\bad \reachOrder \good$.
In fact, it computes annotated orders $(\mathcal{A}, \reachOrderEq^\mathcal{A})$ where $\mathcal{A}$ is a set of assumptions of the form $s \reachOrderEq s'$.
At this stage, the assumptions are not relevant and not used; they become relevant in Sect.~\ref{subsec:makingassumptions}.
The algorithm uses a \textsf{Queue} storing triples consisting of 1) annotations, 2) the order so far, and 3) the remaining states to be processed. 
The queue is initialised (l.~\ref{alg:line:latticeConstruction:initial}) with no annotations, the order $\bad \reachOrder \good$, and the remaining states.
In each iteration, an order is taken from the queue.
If all states are processed, then the order is completed (l.~\ref{alg:line:latticeConstruction:finished}).
Otherwise, some state $s$ is selected (l.~\ref{alg:line:latticeConstruction:selectstate}) to process, and after a possible extension, the queue is updated with the extended order  (l.~\ref{alg:line:latticeConstruction:queueupdate}).
The states are selected in reverse topological order.
Thus, when considering state $s$, all states in $\Succ(s)$ have been considered before.
Using Lem.~\ref{lem:reachrelativesucc}, either $s$ belongs to an already existing equivalence class (l.~\ref{alg:line:latticeConstruction:existing}), or it can be added between some other states (l.~\ref{alg:line:latticeConstruction:addBetween}).
In both cases, the RO-graph of the order $\reachOrderEq$ is extended (where l.~\ref{alg:line:latticeConstruction:existing} uses the extension of $\reachOrderEq$ to equivalence classes).
As assumptions are not used, Alg.~\ref{alg:latticeConstruction} in fact computes a single reachability order; it runs linear in the number of transitions.
\begin{algorithm}[t]
	\caption{Construction of an \rograph}
	\label{alg:latticeConstruction}
	\begin{algorithmic}[1]
	
		\REQUIRE Acyclic pMC $\pDtmcInit$
		\ENSURE Result = a set of annotated orders $\reachOrderEq^\mathcal{A}$ (represented as their RO-graph)
		\STATE Result $\gets$ $\emptyset$, Queue $\gets$ $\left(\mathcal{A} : \emptyset, \reachOrder \, : \{ (\bad,\good) \}, S' : S \setminus \{ \good, \bad \}\right)$\label{alg:line:latticeConstruction:initial}
			\WHILE {Queue not empty}
				\STATE $\mathcal{A}, \reachOrderEqAssumption, S'$ $\gets$ Queue.pop()
				\IF{$S' = \emptyset$}
					\STATE Result $\gets$ Result $\cup$ $\{ (\mathcal{A}, \reachOrderEqAssumption) \}$.\label{alg:line:latticeConstruction:finished}
				\ELSE
				\STATE select $s \in S'$ with $s$ topologically last\label{alg:line:latticeConstruction:selectstate}
				\IF{$\exists s' \in \Succ(s)$ s.t.\  $\Succ(s) \subseteq [s']$}
				\label{alg:line:latticeConstruction:beginIf}
				\STATE extend RO-graph($\reachOrderEqAssumption$) with:  $s \equiv \Succ(s)$ \label{alg:line:latticeConstruction:existing}
				\ELSE
				\STATE extend RO-graph($\reachOrderEqAssumption$) with all: \\ $s \reachOrderAssumption \min\UB(\Succ(s))$ and $\max\LB(\Succ(s)) \reachOrderAssumption s$ \label{alg:line:latticeConstruction:addBetween}
				\ENDIF \label{alg:line:latticeConstruction:endIf}
				\STATE Queue.push($\mathcal{A}, \reachOrderEqAssumption, S' \setminus \{ s \}$)\label{alg:line:latticeConstruction:queueupdate}
				\ENDIF
							\ENDWHILE\label{alg:line:latticeConstruction:endFor}
		\RETURN Result
	\end{algorithmic}
\end{algorithm}
\begin{lemma}
	Algorithm~\ref{alg:latticeConstruction} returns a set with one reachability order.
\end{lemma}
Even if there exists a sufficient reachability order for region $R$, Alg.~\ref{alg:latticeConstruction} might not find such an order, as the algorithm does not take into account $R$ at all --- it is purely graph-based. 
Alg.~\ref{alg:latticeConstruction} does obtain a sufficient reachability order if for all (parametric) states $s\in S$, $\Succ(s)$ is totally ordered by the computed $\reachOrderEq$.

\subsection{Making and discharging assumptions}
\label{subsec:makingassumptions}

Next, we aim to locally refine our RO-graph to obtain sufficient reachability orders.
Therefore, we exploit the annotations (called \emph{assumptions}) that were ignored so far.
Recall from Def.~\ref{def:sufficientRO} that a reachability order is not sufficient at a parametric state $s$, if its successors $s_1$ and $s_2$, say, are not totally ordered. 
We identify these situations while considering $s$ in Alg.~\ref{alg:latticeConstruction}.
We then continue as if the ordering of $s_1$ and $s_2$ is known.
By considering \emph{all} possible orderings of $s_1$ and $s_2$, we remain sound.
The fact that parametric states typically have only two direct successors (as most pMCs are simple~\cite{DBLP:conf/atva/CubuktepeJJKT18, DBLP:conf/uai/Junges0WQWK018}) limits the number of orders.

\tikzstyle{snakeline} = [decorate, decoration={pre length=0.1cm,
                         post length=0.1cm, snake, amplitude=.4mm,
                         segment length=1mm}, gray, ->]

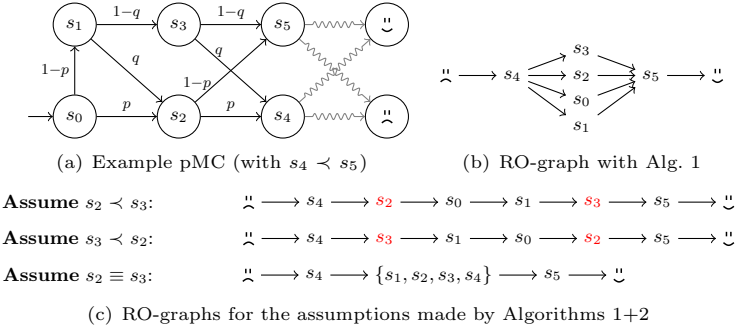
\begin{figure}[t]
\centering
\subfigure[Example pMC (with $s_4 \reachOrder s_5$)]{
\begin{tikzpicture}[scale=0.8, every node/.style={scale=0.8}, node distance=0.8cm]
	\node[state, initial, initial text=] (s0) {$s_0$};
	
	\node[state, above=of s0, yshift=-0.2cm] (s1) {$s_1$};
	
	\node[state, right=of s0] (s2) {$s_2$};
	\node[state, above=of s2, yshift=-0.2cm] (s3) {$s_3$};
	\node[state, right=of s2] (s4) {$s_4$};
	\node[state, right=of s3] (s5) {$s_5$};
	
	\node[state, right=of s5] (g) {$\good$};
	\node[state, right=of s4] (b) {$\bad$};
	
	\draw[->] (s0) edge node[auto, scale=0.8] {$1{-}p$} (s1);
	\draw[->] (s0) edge node[auto, scale=0.8] {$p$} (s2);
	\draw[->] (s1) edge node[auto, scale=0.8] {$q$} (s2);
	\draw[->] (s1) edge node[auto, scale=0.8] {$1{-}q$} (s3);
	\draw[->] (s2) edge node[auto, scale=0.8] {$p$} (s4);
	\draw[->] (s2) edge node[left,pos=0.3, scale=0.8] {$1{-}p$} (s5);
	\draw[->] (s3) edge node[right,pos=0.2, scale=0.8] {$q$} (s4);
	\draw[->] (s3) edge node[auto, scale=0.8] {$1{-}q$} (s5);
	
	\draw[snakeline] (s4) -- (g);
	\draw[snakeline] (s4) -- (b);
	\draw[snakeline] (s5) -- (g);
	\draw[snakeline] (s5) -- (b);
	
\end{tikzpicture}
\label{fig:reqassum:pmc}
}
\subfigure[RO-graph with Alg.~\ref{alg:latticeConstruction}]{
\begin{tikzpicture}[scale=0.8, every node/.style={scale=0.8}, node distance=0.3cm]
	\node (g) {$\good$};
	\node (good) [left=0.5cm of g] {$s_5$};
	\node (s2) [left=0.5cm of good] {$s_2$};
	\node (s3) [above=of s2,yshift=-0.35cm] {$s_3$};
	\node (s0) [below=of s2,yshift=0.35cm] {$s_0$};
	\node (s1) [below=of s0,yshift=0.35cm] {$s_1$};
	\node (bad) [left=0.5cm of s2] {$s_4$};
	\node (b) [left=0.5cm of bad] {$\bad$};
	
	\draw[<-] (g) -- (good);
	
	\draw[<-] (good) -- (s1);
	\draw[<-] (s1) -- (bad);
	\draw[<-] (good) -- (s2);
	\draw[<-] (s2) -- (bad);
	\draw[<-] (good) -- (s3);
	\draw[<-] (s3) -- (bad);
	\draw[<-] (good) -- (s0);
	\draw[<-] (s0) -- (bad);
	
	\draw[<-] (bad) -- (b);
	
	\end{tikzpicture}
\label{fig:reqassum:noassum}
}
\subfigure[RO-graphs for the assumptions made by Algorithms~\ref{alg:latticeConstruction}+\ref{alg:assumptions:naive}]{
\begin{tikzpicture}[scale=0.8, every node/.style={scale=0.8}, node distance=0.5cm]
	\node (gV1) {$\good$};
	\node (goodV1) [left=of gV1] {$s_5$};
	\node (s3V1) [left=of goodV1] {\color{red}{$s_3$}};
	\node (s1V1) [left=of s3V1] {$s_1$};
	\node (s0V1) [left=of s1V1] {$s_0$};
	\node (s2V1) [left=of s0V1] {\color{red}{$s_2$}};
	\node (badV1) [left=of s2V1] {$s_4$};
	\node (bV1) [left=of badV1] {$\bad$};	
	\node (a1V1) [left=1cm of bV1] {\textbf{Assume} $s_2 \reachOrder s_3$:};

	\draw[<-] (gV1) -- (goodV1);
	\draw[<-] (goodV1) -- (s3V1);
	\draw[<-] (s3V1) -- (s1V1);
	\draw[<-] (s1V1) -- (s0V1);
	\draw[<-] (s0V1) -- (s2V1);
	\draw[<-] (s2V1) -- (badV1);
	\draw[<-] (badV1) -- (bV1);

	\node (gV2) [below=of gV1, yshift=0.5cm] {$\good$};
	\node (goodV2) [left=of gV2] {$s_5$};
	\node (s3V2) [left=of goodV2] {\color{red}{$s_2$}};
	\node (s1V2) [left=of s3V2] {$s_0$};
	\node (s0V2) [left=of s1V2] {$s_1$};
	\node (s2V2) [left=of s0V2] {\color{red}{$s_3$}};
	\node (badV2) [left=of s2V2] {$s_4$};
	\node (bV2) [left=of badV2] {$\bad$};
	\node (a1V2) [left=1cm of bV2] {\textbf{Assume} $s_3 \reachOrder s_2$:};
	
	\draw[<-] (gV2) -- (goodV2);
	\draw[<-] (goodV2) -- (s3V2);
	\draw[<-] (s3V2) -- (s1V2);
	\draw[<-] (s1V2) -- (s0V2);
	\draw[<-] (s0V2) -- (s2V2);
	\draw[<-] (s2V2) -- (badV2);
	\draw[<-] (badV2) -- (bV2);

	\node (a1V3) [below=of a1V2, yshift=0.5cm] {\textbf{Assume} $s_2 \equiv s_3$:};
	\node (bV3) [right=1cm of a1V3] {$\bad$};
	\node (badV3) [right=of bV3] {$s_4$};
	\node (CV3) [right=of badV3] {$\{ s_1, s_2, s_3, s_4 \}$};
	\node (goodV3) [right=of CV3] {$s_5$};
	\node (gV3) [right=of goodV3] {$\good$};
		
	\draw[<-] (gV3) -- (goodV3);
	\draw[<-] (goodV3) -- (CV3);
	\draw[<-] (CV3) -- (badV3);
	\draw[<-] (badV3) -- (bV3);
\end{tikzpicture}
\label{fig:reqassum:withassum}
}
\vspace{-0.5em}
\caption{Illustrating the use of assumptions}
\end{figure}

\begin{example}
	Consider the pMC in Fig.~\ref{fig:reqassum:pmc}.
	 Assume that Alg.~\ref{alg:latticeConstruction} yields the RO-graph in Fig.~\ref{fig:reqassum:noassum}, in particular $s_4 \reachOrder s_5$.
	Alg.~\ref{alg:latticeConstruction} cannot order the successors of state $s_1$.
	But any region can be partitioned into three (potentially empty) subregions:
	A region with $s_2 \reachOrder s_3$, a region with $s_2 \equiv s_3$, and a region with $s_3 \reachOrder s_2$. 
	We below adapt Alg.~\ref{alg:latticeConstruction} such that, instead of adding $s_1$ between $s_4$ and $s_5$(l.\ \ref{alg:line:latticeConstruction:addBetween}), we create three copies of the reachability order.
	In the copy assuming $s_2 \reachOrder s_3$ we can order $s_1$ as in Fig.~\ref{fig:reqassum:withassum}.
	The other copies reflect $s_3 \reachOrder s_2$ and $s_2 \equiv s_3$, respectively.
\end{example}
Below, we formalise and automate this.
Let $\mathcal{A} = (\mathcal{A}_\prec, \mathcal{A}_\equiv)$ be a pair of sets of assumptions such that $(s,t) \in \mathcal{A}_\prec$ means $s \reachOrder t$ while $(s,t) \in \mathcal{A}_\equiv$ means $s \equiv t$.

\begin{definition}[Order with assumptions]
\label{def:orderwithassumptions}
Let $\reachOrderEq$ be a reachability order,
and $\mathcal{A} = (\mathcal{A}_\prec, \mathcal{A}_\equiv)$ a pair with \emph{assumptions} $ \mathcal{A}_\prec, \mathcal{A}_\equiv \subseteq S \times S$.
Then $(\reachOrderEq^{\mathcal{A}}, \mathcal{A})$ is called an \emph{order with assumptions} where
$
\reachOrderEq^{\mathcal{A}} \ = \ \bigl( \reachOrderEq \, \cup \, \mathcal{A}_\prec \, \cup \, \mathcal{A}_\equiv \bigr)^*.
$
\end{definition}
The next result asserts that the pre-order $\reachOrderEqAssumption$ is a reachability order if all assumptions conform to the ordering of the reachability probabilities.
\begin{lemma}
If assumptions $\mathcal{A} = (\mathcal{A}_\prec, \mathcal{A}_\equiv)$ satisfy:
\[
\begin{array}{rcl}
(s,t) \in \mathcal{A}_\prec & \quad \text{implies} \quad &
\forall \vec{u} \in R.~ \sol[s]{T}{\pdtmc}(\vec{u}) < \sol[t]{T}{\pdtmc}(\vec{u}), \text{and} \\[1ex]
(s,t) \in \mathcal{A}_\equiv & \quad \text{implies} \quad &
\forall \vec{u} \in R.~ \sol[s]{T}{\pdtmc}(\vec{u}) = \sol[t]{T}{\pdtmc}(\vec{u}),  \end{array}
\]
then $\reachOrderEqAssumption$ is a reachability order, and we call $\mathcal{A}$ \emph{(globally) valid}.
\end{lemma}
Algorithm~\ref{alg:assumptions:naive} adds assumptions to the reachability order. 
It comes before Line~\ref{alg:line:latticeConstruction:addBetween} of Alg.~\ref{alg:latticeConstruction}. 
If the reachability order $\reachOrderEqAssumption$ contains two incomparable successors $s_1$ and $s_2$ of state $s$, we make three different assumptions: In particular, we assume either
$s_1 \reachOrderAssumption s_2$, $s_2 \reachOrderAssumption s_1$, or $s_1 \equiv^\mathcal{A} s_2$.
We then put the updated orders in the queue (without having processed state $s$).
As the states $s_1, s_2$ were incomparable, the assumptions are new and do not contradict with the order so far. 

\begin{algorithm}[t]
	\caption{Assumption extension (put before l.~\ref{alg:line:latticeConstruction:addBetween} in Alg.~\ref{alg:latticeConstruction}).}
	\label{alg:assumptions:naive}
	\begin{algorithmic}[1]
		\IF{$\reachOrderEqAssumption$ is not a total order for $\Succ(s)$}			
			\STATE pick $s_1, s_2 \in \Succ(s)$ s.t.\ neither $s_1 \reachOrderEqAssumption s_2$ nor $s_2 \reachOrderEqAssumption s_1$
			\STATE Queue.push($(\mathcal{A}_\prec \cup \{ (s_1,s_2) \}, \mathcal{A}_\equiv), \reachOrderEqAssumption\text{ extended with }s_1 \prec s_2, S'$) 
			\STATE
			Queue.push($(\mathcal{A}_\prec \cup \{ (s_2,s_1)  \}, \mathcal{A}_\equiv), \reachOrderEqAssumption\text{ extended with }s_1 \equiv s_2, S'$) 
			\STATE Queue.push($(\mathcal{A}_\prec, \mathcal{A}_\equiv \cup \{ (s_1,s_2) \}), \reachOrderEqAssumption\text{ extended with }s_2 \prec s_1, S'$) 
			\STATE \textbf{continue}
		\ENDIF
	\end{algorithmic}
\end{algorithm}

The algorithm does not remove states from the queue if their successors are not totally ordered. Consequently, we have:
\begin{theorem}
\label{theorem-correctness-algoneandtwo}
For every order with assumptions $(\reachOrderEqAssumption, \mathcal{A})$ computed by \linebreak Algorithm~\ref{alg:latticeConstruction}+\ref{alg:assumptions:naive}. Then: if $\reachOrderEqAssumption$ is a reachability order, then it is sufficient.
\end{theorem}

\ownsubsection{Discharging assumptions}
\label{sec:discharge}
Algorithm~\ref{alg:latticeConstruction}+\ref{alg:assumptions:naive} yields a set of orders. 
By Thm.~\ref{theorem-correctness-algoneandtwo}, each order $\reachOrder^\mathcal{A}$ is a (proper) reachability order if the assumptions in $\mathcal{A}$ are valid.

The following result states that the assumptions can sometimes be ignored.
\begin{theorem}
\label{thm:allordersmonotonic}
If all orders computed by Algorithm~\ref{alg:latticeConstruction}+\ref{alg:assumptions:naive} are witnesses for a parameter to be monotonic increasing (decreasing), then the parameter is indeed monotonic increasing (decreasing).
\end{theorem}
This can be seen as follows.
Intuitively, a region $R$ is partitioned into (possibly empty) regions $R_\mathcal{A}$ for each possible set of assumptions $\mathcal{A}$. 
If on each region $R_\mathcal{A}$ the order $\reachOrderEqAssumption$ is a witness for monotonicity (and all witnesses agree on whether the parameter is $\monIncr{}$ or $\monDecr{}$), then the parameter is monotonic on $R$.

If Thm.~\ref{thm:allordersmonotonic} does not apply, we establish whether or not assumptions are valid on $R$ in an on-the-fly manner, as follows:
Let $(\reachOrderEqAssumption, \mathcal{A})$ be the current order, and suppose we want
to check whether $s_1 \reachOrder s_2$ is a new assumption.
If the outcome is $s_1 \reachOrder s_2$, then we extend the RO-graph with $s_1 \reachOrder s_2$, do not add this assumption, and ignore the possibilities $s_1 \equiv s_2$ and $s_2 \reachOrder s_1$.
If $s_1 \not\reachOrderEq s_2$, we do not assume $s_1 \reachOrder s_2$ (and ignore the corresponding order). 
Both cases prune the number of orders.
In case of an inconclusive result, $s_1 \reachOrder s_2$ is added to $\mathcal{A}_\prec$.

To check whether $s_1 \reachOrder s_2$ we describe three techniques.

\smallskip\noindent\emph{Using a local NLP.}
The idea is to locally (at $s_1$ and $s_2$) consider the pMC and its characterising non-linear program (NLP)~\cite{BK08,DBLP:conf/tacas/BartocciGKRS11,DBLP:conf/cav/DehnertJJCVBKA15,DBLP:conf/atva/CubuktepeJJKT18}, together with the inequalities encoded by $\reachOrderEqAssumption$.
To refute an assumption to be globally valid, a single instantiation $\vec{u}$ refuting the assumption suffices. 
This suggests to let a solver prove the absence of such an instantiation $\vec{u}$ by considering a fragment of the pMC (see Example~\ref{ex:NLP}, Appendix~\ref{app:completealg}). 
If successful, the assumption is globally valid. Otherwise, we don't know: the obtained instantiation $\vec{u}$ might be spurious.

\smallskip\noindent\emph{Using model checking.}
This approach targets to cheaply disprove assumptions.
We sample the parameter space at suitable points (as in e.g.\cite{DBLP:conf/tase/ChenHHKQ013,DBLP:journals/jss/CalinescuCGKP18}),
and reduce the amount of solver runs, similar to~\cite{DBLP:conf/cav/DehnertJJCVBKA15}.
In particular, we instantiate the pMC with instantiations $\vec{u}$ from a set $U$, and evaluate the (parameter-free) MC $\pdtmc[\vec{u}]$ via standard model checking. 
This sampling yields reachability probabilities $\reachPrT[s][{\pdtmc[\vec{u}]}]$ for every state $s$, 
and allows to disprove an assumption, say $s_1 \reachOrder s_2$, by merely looking up whether $\reachPrT[s_1][{\pdtmc[\vec{u}]}] \geq \reachPrT[s_2][{\pdtmc[\vec{u}]}]$ for some $u \in U$. 

\smallskip\noindent\emph{Using region checking.}
Region verification procedures (e.g. parameter lifting~\cite{DBLP:conf/atva/QuatmannD0JK16}) consider a region $R$, and obtain for each state $s$ an interval $[a_s,b_s]$ s.t. 
\linebreak $\reachPrT[s][{\pdtmc[\vec{u}]}] \in [a_s,b_s]$ for all $\vec{u} \in R$. 
Assumption $s_1 \reachOrder s_2$ can be proven by checking $b_{s_1} \leq a_{s_2}$.

\subsection{Treating cycles}
\label{subsec:cycles}
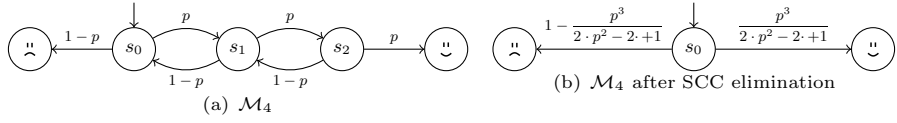
\begin{figure}[t]
	\centering
\subfigure[$\dtmc_4$]{
		\begin{tikzpicture}[scale=0.8, every node/.style={scale=0.8}, node distance=0.8cm, baseline]
		\node[state, initial above]           (q_0)                {$s_0$};
		\node[state]           (q_1) [right=of q_0] {$s_1$};
		\node[state]           (q_2) [right=of q_1] {$s_2$};
		
		\node[state]           (bad) [left=of q_0] {\bad};
		\node[state]           (good) [right=of q_2] {\good};
		
		\path[->] (q_0)   edge node [above,scale=0.8] {$1-p$}  (bad)
		(q_0)   edge [bend left] node [above,scale=0.8] {$p$}  (q_1)
		(q_1)   edge [bend left] node [above,scale=0.8] {$p$}  (q_2)
		(q_1)   edge [bend left]node [below,scale=0.8] {$1-p$}  (q_0)
		(q_2)   edge [bend left]node [below,scale=0.8] {$1-p$}  (q_1)
		(q_2)   edge node [above,scale=0.8] {$p$}  (good);
		\end{tikzpicture}
	\label{fig:scc}
}
\subfigure[$\dtmc_4$ after SCC elimination]{
	\begin{tikzpicture}[scale=0.8, every node/.style={scale=0.8}, node distance=1.8cm, baseline]
	\node[state, initial above]           (q_0)                {$s_0$};
	
	\node[state]           (bad) [left=of q_0] {\bad};
	\node[state]           (good) [right=of q_0] {\good};
	
	\path[->] (q_0)   edge node [above,scale=0.8] {$1-\dfrac{p^3}{2\cdot p^2 -2\cdotp+1}$}  (bad)
	(q_0)   edge [] node [above,scale=0.8] {$\dfrac{p^3}{2\cdot p^2 -2\cdotp+1}$}  (good);
	\end{tikzpicture}
	\label{fig:scc-elimination}
}

\caption{An example pMC consisting of a single SCC}

\end{figure}
So far, we considered acyclic pMCs. We use two techniques to treat cycles.

\smallskip\noindent\emph{SCC elimination}~\cite{DBLP:conf/qest/JansenCVWAKB14} contracts each SCC into a set of states, one for each entry state of the SCC. 
Fig.~\ref{fig:scc-elimination} shows the pMC of Fig.~\ref{fig:scc} after SCC elimination.

\smallskip\noindent\emph{Cycle-breaking.}
If SCC elimination is not viable, we use an alternative. 
The following analogue to Lem.~\ref{lem:reachrelativesucc} is insightful and tailored to simple pMCs.
\begin{restatable}{lemma}{lemForwardreasoning}
	\label{lem:forwardreasoning}
	For any state $s$ with $\Succ(s) = \{ s_1, s_2 \}$ the following holds:
	\begin{inparaenum}
	\item if $s_1 \equiv s$, then $s_2 \equiv s$,
	\item if $s_1 \reachOrder s$, then $s \reachOrder s_2$,
	\item if $s \reachOrder s_1$, then $s_2 \reachOrder s$.	
	\end{inparaenum}
\end{restatable}
\noindent This suggests to take state $s$ on a cycle and insert it into the RO-graph computed so far, which is always (trivially) possible, and then adding further states using Lem.~\ref{lem:forwardreasoning}.
We illustrate this idea by an example.
\begin{example}
Reconsider Fig.~\ref{fig:scc}.
Lem.~\ref{lem:reachrelativesucc} does not give rise to extending the trivial order $\bad \reachOrder \good$.
To treat the cycle, one of the states $s_0$, $s_1$ or $s_2$ is to be added.
Selecting $s_0$ yields (as for any other state) $\bad \reachOrder s_0 \reachOrder \good$.
To order  $s_1$ or $s_2$, Lem.~\ref{lem:reachrelativesucc} is (still) not applicable.
Using $\bad \reachOrder s_0$, Lem.~\ref{lem:forwardreasoning} applied to $s_1$ yields $s_0 \reachOrder s_1$.
For $s_2$, we obtain $s_1 \reachOrder s_2$ in a similar way.
\end{example}
The strategy is thus to successively pick states from a cycle,  insert them into the order $\reachOrderEq$ so far, and continue this procedure until all states on the cycle are covered (by either Lem.~\ref{lem:reachrelativesucc} or \ref{lem:forwardreasoning}).
The extension to Alg.~\ref{alg:latticeConstruction}+\ref{alg:assumptions:naive} is given in Appendix~\ref{app:completealg}.
We emphasise that Theorems~\ref{theorem-correctness-algoneandtwo} and \ref{thm:allordersmonotonic} also apply to this extension.

It remains to discuss: how to decide which states to select on a cycle? 
This is done heuristically. 
A good heuristic selects states that probably lead to cycle ``breaking''. 
The essential criteria that we empirically determined are: take cycles in SCCs that are at the front of the reverse topological ordering of SCCs, and prefer states with successors outside the SCC (as in the above example).

\section{Experimental Evaluation}
We realised a prototype of the algorithm from Sect.~\ref{sec:automatic} on top of \storm~(v1.3)~\cite{DBLP:conf/cav/DehnertJK017} and evaluated two questions.
To that end, we took \emph{all} ten benchmarks sets with pMCs and non-trivial reachability properties from the \param website~\cite{paramwebsite}, and from~\cite{DBLP:conf/tacas/HartmannsKPQR19}, and~\cite{DBLP:journals/iandc/Chatzieleftheriou18}. 
The benchmark sets \textsf{egl}\cite{DBLP:conf/qest/KwiatkowskaNP12}, \textsf{craps}\cite{BK08}, \textsf{nand}\cite{DBLP:journals/tcad/NormanPKS05} and \textsf{herman}~\cite{DBLP:journals/ipl/Herman90,DBLP:journals/fac/KwiatkowskaNP12} are not monotonic. 
Their non-monotonicity can be shown by uniformly taking 100 samples on the parameter space. 
The benchmark \textsf{haddad-monmege}~\cite{DBLP:journals/tcs/HaddadM18} contains only a single non-sink state after preprocessing, it is trivially monotonic.
All experiments ran on a MacBook ME867LL/A. 
We use a 12 GB memory-out (MO), and a 4h time-out (TO).

\ownsubsection{Can the algorithm determine monotonicity on the benchmarks?}
\label{sec:exp:aut}
We consider the performance of the proposed algorithm.
First and foremost, for all six benchmark sets with monotonic parameters, the algorithm automatically and without user interference determines monotonicity.

\begin{table}[t]
\centering
	\scriptsize
	\caption{Automatically inferring monotonicity}
	\vspace{-0.5em}
	\scalebox{0.92}{
	\begin{tabular}{cc|ccrr||c||r|r|r}
		\hline
		benchmark & instance & A/C & $|V|$ & \#states & \#trans & monotonic &  \shortstack{model\\ building} &  \shortstack{mon.\\check} &  \shortstack{sol.\\func.} \\
		\hline
		\hline
		\multirow{3}{*}{\shortstack{\textsf{brp}\\\cite{DBLP:conf/papm/DArgenioJJL01}}} & (2,16) & \multirow{3}{*}{A} & \multirow{3}{*}{2} & 613 & 803 & \multirow{3}{*}{$\monDecrNoR[pK]{}$, $\monDecrNoR[pL]{}$}  & $<1$ & $\mathbf{<1}$ &  $\mathbf{<1}$   \\
		& (10,2048) &  &  & 45059 & 90115 & & 6 & \textbf{1} & \textbf{MO}\\
		& (15,4096) &  &  & 131075 & 262147 &  & 16 & \textbf{13}& \textbf{MO}\\
		\hline
		\multirow{3}{*}{\shortstack{\textsf{crowds}\\\cite{DBLP:journals/tissec/ReiterR98}}} 
		& (5,6) & \multirow{3}{*}{C} & \multirow{3}{*}{2} & 18817 & 32677 & \multirow{3}{*}{$\monIncrNoR[badC]{}$, $\monIncrNoR[pF]{}$} & $<1$ & \textbf{1} & $\mathbf{<1}$ \\
		& (10,6) &  &  & 352535 & 722015 &  & 6 &\textbf{1}& $\mathbf{<1}$ \\
		& (20,6) &  &  & 10633591 & 27151191 &  & 232 &\textbf{1} &  $\mathbf{<1}$ \\
		\hline
		\multirow{3}{*}{\shortstack{\textsf{gambler}\\\cite{DBLP:journals/iandc/Chatzieleftheriou18}}} 
		& (14800,1480) & \multirow{3}{*}{C} & \multirow{3}{*}{1} & 16281 & 32560 & \multirow{3}{*}{$\monIncrNoR[p]{}$} & $<1$& \textbf{1} &\textbf{TO}\\
		& (29600,2960) &  &  & 32561 & 65120 & & $<1$& \textbf{6} & \textbf{TO}  \\
		& (59200,5920) &  &  & 65121 & 130240 & & 2 & \textbf{21} &\textbf{TO} \\
		\hline
\multirow{3}{*}{\shortstack{\textsf{mes. auth.}\\\cite{DBLP:conf/icse/FilieriGT11}}}
		& (3840) & \multirow{3}{*}{A} & \multirow{3}{*}{2} & 19201 & 30720 &\multirow{3}{*}{$\monIncrNoR[p]{}$, $\monIncrNoR[q]{}$} & 3 & $\mathbf{<1}$ & $\mathbf{<1}$  \\
		& (7680) &  &  & 38401 & 61440 &  & 4 & $\mathbf{<1}$  & $\mathbf{<1}$ \\
		& (15360) &  &  & 76801 & 122880 & & 4 & $\mathbf{<1}$  & $\mathbf{<1}$ \\
		\hline
		\multirow{3}{*}{\shortstack{\textsf{zeroconf}\\\cite{DBLP:conf/dsn/BohnenkampSHV03}}}
		& (6400) & \multirow{3}{*}{C} & \multirow{3}{*}{2} & 6404 & 12805 & \multirow{3}{*}{$\monIncrNoR[p]{}$, $\monIncrNoR[q]{}$}  & $<1$ &$\mathbf{<1}$ &\textbf{1090} \\
		& (25600) &  &  & 25604 & 51205 & & $<1$ &$\mathbf{<1}$  &  \textbf{TO} \\
		& (102400) &  &  & 102404 & 204805 & & 3 & \textbf{3} & \textbf{TO} \\
		\hline
	\end{tabular}
	}
	\label{tab:aut}
	
\end{table}

Table~\ref{tab:aut} presents details: it lists the benchmark and their instances. We then list
whether the pMC is acyclic (A) or cyclic (C), the number $|V|$ of parameters, and the size of the pMC.
The column \emph{monotonic} gives the obtained results for the pMC parameters. 
\emph{Model building} includes the time for construction, default preprocessing by \storm, and bisimulation minimisation.
\emph{Mon.\ check} shows timings for inferring monotonicity from the built model.
To place these numbers in perspective, column \emph{sol.\ func} shows the time to obtain the solution function from the built model by \storm (default settings, same preprocessing, based on the implementation in \cite{DBLP:conf/cav/DehnertJJCVBKA15}). 
These times are a lower bound on the time to show monotonicity via the solution function. 
Timings for taking the derivative and analysing this derivative via an SMT solver are omitted (but significant).

The proposed method quickly determines monotonicity. 
For (only) \textsf{crowds}, the method applies an essential SCC elimination on the various smaller SCCs. 
For the available benchmarks, the method computes a single reachability order.
Unsurprisingly, the method is orders of magnitude faster and scales better than obtaining monotonicity from the solution function. 
Naturally, our algorithm cannot establish monotonicity on all cases.
The algorithm has difficulties handling subregions on non-monotonic benchmarks.
Take \textsf{Herman}: The solution function has (up to) three local extrema~\cite{DBLP:journals/fac/KwiatkowskaNP12}. 
On subregions, however, the graph structure easily induces inconclusive orders. 
A tighter integration with region verification, partially applied state elimination, or using a notion of multi-step local monotonicity (used in the proof of Thm.~\ref{thm:monIncr}) are avenues for improvement.

\ownsubsection{Does monotonicity allow for faster parameter synthesis?}
We consider three variants of parameter synthesis in the presence of monotonicity: 

\smallskip
\noindent\emph{Feasibility:} i.e., \emph{is there an instantiation for which a specification $\varphi$ is satisfied?} becomes mostly trivial in the presence of monotonicity. For simple pMCs, a single parameter-free MC evaluation suffices, which is clearly superior to other---typically sampling-based---approaches~\cite{DBLP:conf/tase/ChenHHKQ013,DBLP:conf/atva/CubuktepeJJKT18}.

\smallskip\noindent\emph{Region verification:} i.e., \emph{do all parameter values within a region satisfy $\varphi$?} is similarly trivialised for regions given as linear polyhedra. 
For our benchmarks, PLA~\cite{DBLP:conf/atva/QuatmannD0JK16}---approximating region verification by MDP model checking---is very competitive. 
In particular, PLA does not over-approximate on locally monotonic pMCs, and needs no refinements (the reverse does not hold: even for tight bounds, one cannot infer monotonicity with PLA).
Thus, whereas sampling checks a single MC, PLA checks one MDP. 
Typically, the MC can be checked $\sim$20\% faster. 

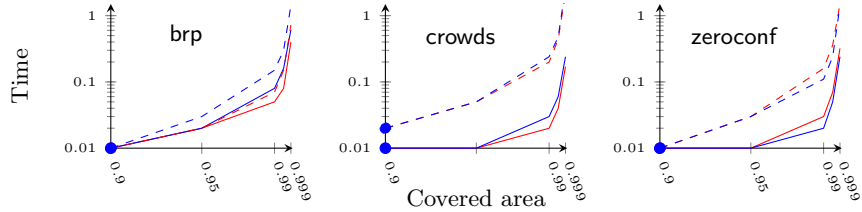
\begin{figure}[t]
	\centering
		\begin{tikzpicture}[baseline]
				\node (a) at (1,1.5) { \textsf{brp} };
		\begin{axis}[
		width=4cm, height=3.5cm, ymin=0, xmin=0.9, ymax=1.5, xmax=1, ymode=log,
				axis x line=bottom, axis y line=left,
		x label style={at={(axis description cs:0.5,0.05)},anchor=north},
		y label style={at={(axis description cs:0.1,0.5)},anchor=south},
		xtick={0.9, 0.95, 0.99, 0.999},
		xticklabels={0.9, 0.95, 0.99, 0.999},
		ytick={0.01, 0.1, 1},
		yticklabels={0.01, 0.1, 1},
		ylabel=Time,
		yticklabel style={font=\tiny}, xticklabel style={rotate=290, anchor=west,  font=\tiny},
		mark repeat={1000},
		legend pos=south east, legend style={font=\tiny}]
		\addplot[color=red,mark=x]
		table [col sep=semicolon, y=mon1, x=coverage]
		{plots/brp.csv};
		\addplot[color=blue,mark=*]
		table [col sep=semicolon, y=mon2, x=coverage]
		{plots/brp.csv};
		\addplot[color=red,mark=x,style=dashed]
		table [col sep=semicolon, y=pla1, x=coverage]
		{plots/brp.csv};
		\addplot[color=blue,mark=*, style=dashed]
		table [col sep=semicolon, y=pla2, x=coverage]
		{plots/brp.csv};
		\end{axis}
		\
		\end{tikzpicture}
		\label{fig:plot:brp}
		\begin{tikzpicture}[baseline]
						\node (a) at (1,1.5) { \textsf{crowds} };
		\begin{axis}[
		width=4cm, height=3.5cm, ymin=0, xmin=0.9, ymax=1.5, xmax=1, ymode=log,
				axis x line=bottom, axis y line=left,
		x label style={at={(axis description cs:0.5,0.05)},anchor=north},
		y label style={at={(axis description cs:0.05,0.5)},anchor=south},
		xtick={0.9, 0.95, 0.99, 0.999},
		xticklabels={0.9, , 0.99, 0.999},
		ytick={0.01, 0.1, 1},
		yticklabels={0.01, 0.1, 1},
		 xlabel=Covered area,
		yticklabel style={font=\tiny}, xticklabel style={rotate=290, anchor=west, font=\tiny},
		mark repeat={1000},
		legend pos=south east, legend style={font=\tiny}]
		\addplot[color=red,mark=x]
		table [col sep=semicolon, y=mon1, x=coverage]
		{plots/crowds.csv};
		\addplot[color=blue,mark=*]
		table [col sep=semicolon, y=mon2, x=coverage]
		{plots/crowds.csv};
		\addplot[color=red,mark=x,style=dashed]
		table [col sep=semicolon, y=pla1, x=coverage]
		{plots/crowds.csv};
		\addplot[color=blue,mark=*, style=dashed]
		table [col sep=semicolon, y=pla2, x=coverage]
		{plots/crowds.csv};
		\end{axis}
		\end{tikzpicture}
		\label{fig:plot:crowds}
		\begin{tikzpicture}[baseline]
						\node (a) at (1,1.5) { \textsf{zeroconf} };
		\begin{axis}[
		width=4cm, height=3.5cm, ymin=0, xmin=0.9, ymax=1.5, xmax=1, ymode=log,
				axis x line=bottom, axis y line=left,
		x label style={at={(axis description cs:0.5,0.05)},anchor=north},
		y label style={at={(axis description cs:0.05,0.5)},anchor=south},
		xtick={0.9, 0.95, 0.99, 0.999},
		xticklabels={0.9, 0.95, 0.99, 0.999},
		ytick={0.01, 0.1, 1},
		yticklabels={0.01, 0.1, 1},
		yticklabel style={font=\tiny}, xticklabel style={rotate=290, anchor=west, font=\tiny},
		mark repeat={1000},
		legend pos=south east, legend style={font=\tiny}]
		\addplot[color=red,mark=x]
		table [col sep=semicolon, y=mon1, x=coverage]
		{plots/zeroconf.csv};
		\addplot[color=blue,mark=*]
		table [col sep=semicolon, y=mon2, x=coverage]
		{plots/zeroconf.csv};
		\addplot[color=red,mark=x,style=dashed]
		table [col sep=semicolon, y=pla1, x=coverage]
		{plots/zeroconf.csv};
		\addplot[color=blue,mark=*, style=dashed]
		table [col sep=semicolon, y=pla2, x=coverage]
		{plots/zeroconf.csv};
		\end{axis}
		\end{tikzpicture}
		\label{fig:plot:zeroconf}
	\caption{\emph{What coverage (x-axis) in how much time (y-axis)}? when using PLA (dotted) or Monotonicity/Sampling (solid) 
		. The two colours indicate two different thresholds. }
	\label{fig:plot}
\end{figure}
\smallskip\noindent\emph{Parameter space partitioning:}
this procedure, implemented in \prophesy, \storm, and \param, \emph{iteratively divides a region into subregions that satisfy $\varphi$ or $\neg\varphi$, respectively}. 
We implemented an alternative prototype based on sampling and exploiting monotonicity: therefore, region splits can be taken much more informed. 
We compared to \storm (using PLA). 
Fig.~\ref{fig:plot} (log-log scale) displays cumulative model-checking runtimes to achieve a given coverage. 
Obtaining a coverage up to 90\% is trivial.
For higher coverage, our method is (on \textsf{crowds} and \textsf{zeroconf}) up to an order of magnitude faster, due to less model-checking calls.
This trend is independent of the threshold.
We see some room for improvement by a more sophisticated selection of samples, and by speeding up the sampling~\cite{DBLP:conf/atva/GainerHS18}.

\section{Related Work}

\noindent\emph{Monotonicity.}
Monotonicity in MCs goes back to Daley~\cite{Daley1968},
aiming to bound stationary probabilities of a stochastically monotone MC by another MC.
These stochastic orderings $\leq_{st}$ require ordered rows in the matrix $\mathcal{P}$, and are quite different from reachability orders. 
MCs can be compared if $\mathcal{P}$ is monotone w.r.t.\ $\leq_{st}$ for all probability vectors.
Such orderings have been used for multi-valued model checking of interval MCs~\cite{DBLP:conf/qest/HaddadP09}, but not applied to pMCs.

\smallskip\noindent\emph{Pre-orders.}
A never-worse relation (NWR) on MDP states~\cite{DBLP:conf/ijcai/BharadwajRPT17,DBLP:conf/fossacs/0001P18} is similar in spirit to reachability orders: states are ordered according to their maximal reachability probabilities but without taking the probabilities into account.
Dependencies between the state probabilities are thus not taken into account.
Like in our setting, computing the NWR is based on the graph structure.
Its usage however is quite different, reducing the size of the MDP prior to model checking.
NWR captures most heuristics to reduce the MDP before linear programming or value iteration.
The coNP-completeness~\cite{DBLP:conf/fossacs/0001P18} indicates that checking this order is simpler than our pre-order unless coETR and coNP coincide.

\smallskip\noindent\emph{Monotonicity in parameter synthesis.}
Parameter lifting~\cite{DBLP:conf/atva/QuatmannD0JK16} exploits a form of local monotonicity to remove parameter dependencies in a pMC.
(A similar observation for continuous-time MCs was made in~\cite{DBLP:conf/cav/BrimCDS13}.)
The resulting monotonic pMC is replaced by an MDP that over-approximates the original pMC.
No efforts are made to determine global monotonicity.
Interval MCs~\cite{DBLP:conf/lics/JonssonL91,DBLP:conf/fossacs/ChatterjeeSH08} lack dependencies, thus all states are locally monotonic (but it remains unclear whether they are monotonically increasing or decreasing).
Monotonicity also affects complexity.
Hutschenreiter \emph{et al.}~\cite{DBLP:journals/corr/abs-1709-02093} recently showed that the complexity of model checking (a monotone fragment of) PCTL on monotonic pMC is lower than PCTL model checking on general pMCs.
They use a very restrictive sufficient criterion for a pMC to be monotonic: This includes none of the pMCs considered in this paper.
Monotonicity has also been considered in the context of model repair.
Pathak \emph{et al.}~\cite{DBLP:conf/nfm/PathakAJTK15} provide an efficient greedy approach to repair monotonic pMCs\footnote{Although monotonicity is not explicitly mentioned in~\cite{DBLP:conf/nfm/PathakAJTK15}.}. 
Recently, Gouberman \emph{et al.}~\cite{GOUBERMAN201932} show that particular perturbations of direct predecessors of  \good or \bad in a continuous-time MC are monotonic in the perturbation factor.

\section{Conclusion and Future Work}
We proposed a method that automatically infers the monotonicity of pMCs from the literature.
To the best of our knowledge, our paper is the first automated procedure for determining monotonicity.
Future work includes a tighter integration with parameter synthesis, and extensions to pMDPs and rewards.

\bibliographystyle{splncs04}
\bibliography{literature}

\begin{thebibliography}{10}
\providecommand{\url}[1]{\texttt{#1}}
\providecommand{\urlprefix}{URL }
\providecommand{\doi}[1]{https://doi.org/#1}

\bibitem{paramwebsite}
\param website (2019), \url{https://depend.cs.uni-saarland.de/tools/param/}

\bibitem{DBLP:conf/srds/AflakiVBKS17}
Aflaki, S., Volk, M., Bonakdarpour, B., Katoen, J.P., Storjohann, A.: Automated
  fine tuning of probabilistic self-stabilizing algorithms. In: {SRDS}. {IEEE}
  CS (2017)

\bibitem{DBLP:reference/mc/BaierAFK18}
Baier, C., de~Alfaro, L., Forejt, V., Kwiatkowska, M.: Model checking
  probabilistic systems. In: Handbook of Model Checking. Springer (2018)

\bibitem{BK08}
Baier, C., Katoen, J.P.: Principles of model checking. {MIT} Press (2008)

\bibitem{DBLP:conf/tacas/BartocciGKRS11}
Bartocci, E., Grosu, R., Katsaros, P., Ramakrishnan, C.R., Smolka, S.A.: Model
  repair for probabilistic systems. In: {TACAS}. {LNCS}, vol.~6605. Springer
  (2011)

\bibitem{DBLP:conf/ijcai/BharadwajRPT17}
Bharadwaj, S., Roux, S.L., P{\'{e}}rez, G.A., Topcu, U.: Reduction techniques
  for model checking and learning in {MDP}s. In: {IJCAI}. ijcai.org (2017)

\bibitem{DBLP:conf/dsn/BohnenkampSHV03}
Bohnenkamp, H.C., van~der Stok, P., Hermanns, H., Vaandrager, F.W.:
  Cost-optimization of the ipv4 zeroconf protocol. In: {DSN}. {IEEE} CS (2003)

\bibitem{DBLP:conf/cav/BrimCDS13}
Brim, L., Ceska, M., Drazan, S., Safr{\'{a}}nek, D.: Exploring parameter space
  of stochastic biochemical systems using quantitative model checking. In:
  {CAV}. {LNCS}, vol.~8044. Springer (2013)

\bibitem{DBLP:journals/jss/CalinescuCGKP18}
Calinescu, R., Ceska, M., Gerasimou, S., Kwiatkowska, M., Paoletti, N.:
  Efficient synthesis of robust models for stochastic systems. J. Syst. Softw.
  \textbf{143} (2018)

\bibitem{DBLP:journals/acta/CeskaDPKB17}
Ceska, M., Dannenberg, F., Paoletti, N., Kwiatkowska, M., Brim, L.: Precise
  parameter synthesis for stochastic biochemical systems. Acta Inf.
  \textbf{54}(6) (2017)

\bibitem{DBLP:conf/fossacs/ChatterjeeSH08}
Chatterjee, K., Sen, K., Henzinger, T.A.: Model-checking omega-regular
  properties of interval {M}arkov chains. In: FoSSaCS. {LNCS}, vol.~4962.
  Springer (2008)

\bibitem{DBLP:journals/iandc/Chatzieleftheriou18}
Chatzieleftheriou, G., Katsaros, P.: Abstract model repair for probabilistic
  systems. Inf. Comput.  \textbf{259}(1) (2018)

\bibitem{DBLP:conf/tase/ChenHHKQ013}
Chen, T., Hahn, E.M., Han, T., Kwiatkowska, M.Z., Qu, H., Zhang, L.: Model
  repair for {M}arkov decision processes. In: {TASE}. {IEEE} (2013)

\bibitem{DBLP:journals/corr/Chonev17}
Chonev, V.: Reachability in augmented interval {M}arkov chains. CoRR
  \textbf{abs/1701.02996} (2017)

\bibitem{DBLP:conf/atva/CubuktepeJJKT18}
Cubuktepe, M., Jansen, N., Junges, S., Katoen, J.P., Topcu, U.: Synthesis in
  p{MDP}s: {A} tale of 1001 parameters. In: {ATVA}. {LNCS}, vol. 11138.
  Springer (2018)

\bibitem{Daley1968}
Daley, D.J.: Stochastically monotone {M}arkov chains. Zeitschrift für
  Wahrscheinlichkeitstheorie und Verwandte Gebiete  \textbf{10} (1968)

\bibitem{DBLP:conf/papm/DArgenioJJL01}
D'Argenio, P.R., Jeannet, B., Jensen, H.E., Larsen, K.G.: Reachability analysis
  of probabilistic systems by successive refinements. In: {PAPM-PROBMIV}.
  {LNCS}, vol.~2165. Springer (2001)

\bibitem{Daws04}
Daws, C.: Symbolic and parametric model checking of discrete-time {M}arkov
  chains. In: Proc.\ of ICTAC. LNCS, vol.~3407. Springer (2004)

\bibitem{DBLP:conf/cav/DehnertJJCVBKA15}
Dehnert, C., Junges, S., Jansen, N., Corzilius, F., Volk, M., Bruintjes, H.,
  Katoen, J.P., {\'{A}}brah{\'{a}}m, E.: Prophesy: {A} probabilistic parameter
  synthesis tool. In: {CAV} {(1)}. {LNCS}, vol.~9206. Springer (2015)

\bibitem{DBLP:conf/cav/DehnertJK017}
Dehnert, C., Junges, S., Katoen, J.P., Volk, M.: A storm is coming: {A} modern
  probabilistic model checker. In: {CAV} {(2)}. {LNCS}, vol. 10427. Springer
  (2017)

\bibitem{DBLP:conf/icse/FilieriGT11}
Filieri, A., Ghezzi, C., Tamburrelli, G.: Run-time efficient probabilistic
  model checking. In: {ICSE}. {ACM} (2011)

\bibitem{DBLP:journals/tse/FilieriTG16}
Filieri, A., Tamburrelli, G., Ghezzi, C.: Supporting self-adaptation via
  quantitative verification and sensitivity analysis at run time. {IEEE} TSE
  \textbf{42}(1) (2016)

\bibitem{DBLP:conf/atva/GainerHS18}
Gainer, P., Hahn, E.M., Schewe, S.: Accelerated model checking of parametric
  {M}arkov chains. In: {ATVA}. {LNCS}, vol. 11138. Springer (2018)

\bibitem{GOUBERMAN201932}
Gouberman, A., Siegle, M., Tati, B.: Markov chains with perturbed rates to
  absorption: Theory and application to model repair. Perf. Ev.  \textbf{130}
  (2019)

\bibitem{DBLP:journals/tcs/HaddadM18}
Haddad, S., Monmege, B.: Interval iteration algorithm for {MDP}s and {IMDP}s.
  Theor. Comput. Sci.  \textbf{735} (2018)

\bibitem{DBLP:conf/qest/HaddadP09}
Haddad, S., Pekergin, N.: Using stochastic comparison for efficient model
  checking of uncertain {M}arkov chains. In: {QEST}. {IEEE} CS (2009)

\bibitem{param_sttt}
Hahn, E.M., Hermanns, H., Zhang, L.: Probabilistic reachability for parametric
  {M}arkov models. Software Tools for Technology Transfer  \textbf{13}(1)
  (2010)

\bibitem{DBLP:conf/tacas/HartmannsKPQR19}
Hartmanns, A., Klauck, M., Parker, D., Quatmann, T., Ruijters, E.: The
  quantitative verification benchmark set. In: {TACAS}. LNCS, vol. 11427.
  Springer (2019)

\bibitem{DBLP:journals/ipl/Herman90}
Herman, T.: Probabilistic self-stabilization. Inf. Process. Lett.
  \textbf{35}(2) (1990)

\bibitem{DBLP:journals/corr/abs-1709-02093}
Hutschenreiter, L., Baier, C., Klein, J.: Parametric {M}arkov chains: {PCTL}
  complexity and fraction-free {G}aussian elimination. In: GandALF. {EPTCS},
  vol.~256 (2017)

\bibitem{DBLP:conf/qest/JansenCVWAKB14}
Jansen, N., Corzilius, F., Volk, M., Wimmer, R., {\'{A}}brah{\'{a}}m, E.,
  Katoen, J.P., Becker, B.: Accelerating parametric probabilistic verification.
  In: {QEST}. {LNCS}, vol.~8657. Springer (2014)

\bibitem{DBLP:conf/lics/JonssonL91}
Jonsson, B., Larsen, K.G.: Specification and refinement of probabilistic
  processes. In: {LICS}. {IEEE} CS (1991)

\bibitem{DBLP:journals/cca/JovanovicM12}
Jovanovic, D., de~Moura, L.: Solving non-linear arithmetic. {ACM} Comm.
  Computer Algebra  \textbf{46}(3/4) (2012)

\bibitem{DBLP:conf/uai/Junges0WQWK018}
Junges, S., Jansen, N., Wimmer, R., Quatmann, T., Winterer, L., Katoen, J.P.,
  Becker, B.: Finite-state controllers of {POMDP}s using parameter synthesis.
  In: {UAI}. {AUAI} Press (2018)

\bibitem{DBLP:conf/lics/Katoen16}
Katoen, J.P.: The probabilistic model checking landscape. In: {LICS}. {ACM}
  (2016)

\bibitem{DBLP:conf/cav/KwiatkowskaNP11}
Kwiatkowska, M.Z., Norman, G., Parker, D.: {PRISM} 4.0: Verification of
  probabilistic real-time systems. In: {CAV}. {LNCS}, vol.~6806. Springer
  (2011)

\bibitem{DBLP:conf/qest/KwiatkowskaNP12}
Kwiatkowska, M.Z., Norman, G., Parker, D.: The {PRISM} benchmark suite. In:
  {QEST}. {IEEE} CS (2012)

\bibitem{DBLP:journals/fac/KwiatkowskaNP12}
Kwiatkowska, M.Z., Norman, G., Parker, D.: Probabilistic verification of
  {H}erman's self-stabilisation algorithm. Formal Asp. Comput.
  \textbf{24}(4-6) (2012)

\bibitem{DBLP:journals/tcad/NormanPKS05}
Norman, G., Parker, D., Kwiatkowska, M.Z., Shukla, S.K.: Evaluating the
  reliability of {NAND} multiplexing with {PRISM}. {IEEE} Trans. on {CAD} of
  Integrated Circuits and Systems  \textbf{24}(10) (2005)

\bibitem{DBLP:conf/nfm/PathakAJTK15}
Pathak, S., {\'{A}}brah{\'{a}}m, E., Jansen, N., Tacchella, A., Katoen, J.P.: A
  greedy approach for the efficient repair of stochastic models. In: {NFM}.
  {LNCS}, vol.~9058 (2015)

\bibitem{DBLP:conf/atva/QuatmannD0JK16}
Quatmann, T., Dehnert, C., Jansen, N., Junges, S., Katoen, J.P.: Parameter
  synthesis for {M}arkov models: Faster than ever. In: {ATVA}. {LNCS},
  vol.~9938 (2016)

\bibitem{DBLP:journals/tissec/ReiterR98}
Reiter, M.K., Rubin, A.D.: Crowds: Anonymity for web transactions. {ACM} Trans.
  Inf. Syst. Secur.  \textbf{1}(1) (1998)

\bibitem{DBLP:conf/fossacs/0001P18}
Roux, S.L., P{\'{e}}rez, G.A.: The complexity of graph-based reductions for
  reachability in {M}arkov decision processes. In: FoSSaCS. {LNCS}, vol. 10803.
  Springer (2018)

\bibitem{DBLP:journals/corr/abs-1904-01503}
Winkler, T., Junges, S., P{\'{e}}rez, G.A., Katoen, J.: On the complexity of
  reachability in parametric markov decision processes. CoRR
  \textbf{abs/1904.01503} (2019)

\end{thebibliography}

\clearpage
\pagebreak

\appendix
\section{Proofs}

\subsection{Proof of Theorem~\ref{lem:complexitymonproblem} and Lem.~\ref{lem:complexity}}
Reductions are many-to-one reductions. The proofs use notions from Sect.~\ref{sec:ro}.

We use the following formal definition of the decision problems:

\begin{itemize}
\item
\textbf{(pMC-V)}  
\emph{Given a pMC $\pdtmc$ with a threshold $\lambda \in [0,1]$, and a graph-preserving region $R$, does $\solWithM[s]{T}{\pdtmc}(\vec{u}) \geq \lambda$ hold for all $\vec{u} \in R$?}
\item 
\textbf{(pMC-RO)}  
\emph{Given a pMC $\pdtmc$ with a graph-preserving region $R$, 
	and two states $s_1$ and $s_2$, does $s_1 \reachOrderEq[R,T] s_2$ hold, where  $\reachOrderEq[R,T]$ denotes the exhaustive reachability order?}
\item
\textbf{(pMC-Mon)}
\emph{Given a pMC $\pdtmc$ with a parameter $p \in \Var$, and a graph-preserving region $R$, does $\monIncr[p]{\pdtmc}$ hold?}
\end{itemize}

\montheorem*
\lemcomplexity*
With that Lem.~\ref{lem:complexity} thus states that pMC-V is reducible to pMC-RO, 
and Theorem~\ref{lem:complexitymonproblem} states that pMC-V is reducible to pMC-Mon.

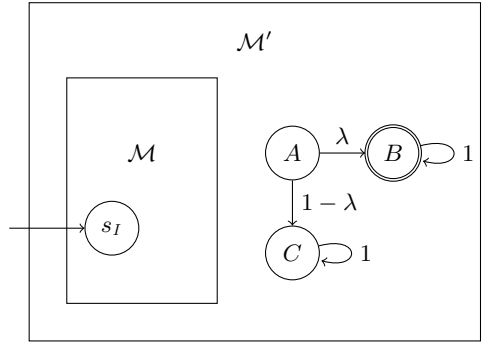
\begin{figure*}[t]
\centering
	\begin{tikzpicture}
		\draw (-0.5, -0.5) rectangle (5.5,4);
		\node[] at (2.5,3.5) {$\pdtmc'$};
		\draw (0,0) rectangle (2,3);
		\node[] at (1,2) {$\pdtmc$};
		
		\node[state, initial, initial distance=1cm] at (0.6,1) (sinit) {$\sinit$};
		
		\node[state] at (3,2) (A) {$A$};
		\node[state, accepting, right=0.6cm of A] (B) {$B$};
		\node[state, below=0.6cm of A] (C) {$C$};
		
		\draw[->] (A) edge node[auto] {$\lambda$} (B);
		\draw[->] (A) edge node[auto] {$1-\lambda$} (C);
		
		\draw[->] (B) edge[loop right] node[auto] {$1$} (B);
		\draw[->] (C) edge[loop right] node[auto] {$1$} (C);
	\end{tikzpicture}
	\caption{Outline of the construction for the proof of Lem.~\ref{lem:complexity}}
	\label{fig:vertoro}
\end{figure*}

\begin{proof}[of Lem.~\ref{lem:complexity}]
Given an instance of pMC-V, we construct an instance of pMC-RO as follows (see also Fig.~\ref{fig:vertoro}):
We construct a new $\pDtmcInit[']$ by taking $\pDtmcInit[]$ and adding an (unconnected) gadget consisting of three states $A, B, C$. Let $B, C$ be sink states. We add a transition from $A$ to $B$ with probability $\lambda$.
Formally:
\begin{itemize}
	 \item $S' = S \uplus \{A, B, C\}$
	 \item  $\sinit' = \sinit$
	 \item $T' = T \cup \{ B \}$
	 \item $\Var' = \Var$
	 \item  \begin{tabular}{l} $\probdtmc'(s,s') =\begin{cases}
	 \probdtmc(s,s') & \text{if }s, s' \in S \\
	 \lambda & \text{if }s=A, s'=B \\
	 1-\lambda & \text{if }s=A, s'=C \\
	 1 & \text{if }s=s'=B \\
	 1 & \text{if }s=s'=C \\ 	
	 0 & \text{otherwise}			
	  \end{cases}$\end{tabular}

\end{itemize}	 
 States $s_1$ and $s_2$ now correspond to $A$ and $\sinit$, respectively.
This transformation clearly is in polynomial time. 

Now, the probability to reach the target from $A$ is $\lambda$ for any parameter instantiation.
Thus, $A \prec \sinit$ iff $\solWithM[s_I]{T}{\pdtmc}(\vec{u}) \geq \lambda$ for all $u \in R$.\qed
\end{proof}

We now turn our attention to Theorem~\ref{lem:complexitymonproblem}. 
We first show:
\begin{lemma}
\label{lem:monishard}
	pMC-RO is polynomially reducible to pMC-Mon.
\end{lemma}

\begin{figure*}[t]
\centering
	\begin{tikzpicture}
		\draw (-2.2,-0.5) rectangle (2.5,4);
		\node[] at (0.15,3.5) {$\pdtmc'$};
		\draw (0,0) rectangle (2,3);
		\node[] at (1,2.3) {$\pdtmc$};
		
		\node[state] at (0.8,1.6) (s1) {$s_1$};
		\node[state] at (0.8,0.5) (s2) {$s_2$};
		
		\node[state, initial, initial distance=1cm] at (-1.2,1) (D) {$D$};

		\draw[->] (D) edge node[pos=0.3,above] {$1-p$} (s1);
		\draw[->] (D) edge node[pos=0.3,below] {$p$} (s2);
		
	\end{tikzpicture}
	\caption{Outline of the construction for the proof of Lem.~\ref{lem:monishard}}
	\label{fig:rotomon}
\end{figure*}
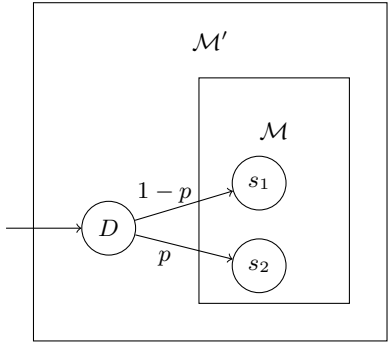

\begin{proof}

	Given an instance of pMC-RO, with a graph-preserving region $R$, 
	and two states $s_1$ and $s_2$, we construct an instance of pMC-Mon as follows.
	For $\pDtmcInit$, we construct $\pDtmcInit[']$ by extending $\pdtmc$, see also Fig.~\ref{fig:rotomon}.
	We introduce a fresh variable $p$ and a fresh initial state $D$, and connect $D$ with $s_1$ with probability $1-p$, and $D$ with $s_2$ with probability $p$. 
Formally:	
\begin{itemize}
	 \item $S' = S \uplus \{D\}$
	 \item $\sinit' = D$
	 \item $T' = T $
	 \item $\Var' = \Var \uplus \{ p \}$
	 \item \begin{tabular}{l} $\probdtmc'(s,s') =\begin{cases}
	 \probdtmc(s,s') & \text{if }s, s' \in S \\
	 1-p & \text{if }s=D, s'=s_1 \\
	 p & \text{if }s=D, s'=s_2 \\
	 0 & \text{otherwise}			
	  \end{cases}$ \end{tabular}s
\end{itemize}	 
 States $s_1$ and $s_2$ now correspond to $A$ and $\sinit$ from the previous reduction proof.
This transformation clearly is in polynomial time. 

First assume $s_1 \reachOrderEq s_2$:
Then, for any instantiation $\vec{u} \in R$ s.t.\ $\solWithM[s_1]{T}{\pdtmc}(\vec{u}) \leq \solWithM[s_2]{T}{\pdtmc}(\vec{u})$. 
Clearly, increasing $p$ then increases $\solWithM[D]{T}{\pdtmc'}(\vec{u})$, so $\pdtmc'$ is monotonic increasing in $p$ on $R$.

Now assume $s_1 \not\reachOrderEq s_2$:
 Then, there is an instantiation $\vec{u} \in R$ s.t.\ $\solWithM[s_1]{T}{\pdtmc}(\vec{u}) \geq \solWithM[s_2]{T}{\pdtmc}(\vec{u})$. 
 This even holds when varying $p$, as $\solWithM[s_1]{T}{\pdtmc}(\vec{u})$, $\solWithM[s_2]{T}{\pdtmc}(\vec{u})$ are independent of $p$. 
 We observe that for this instantiation $\vec{u}$ increasing $p$ decreases the value of $\solWithM[D]{T}{\pdtmc'}(\vec{u})$, so $\pdtmc'$ is not monotonic increasing in $p$.

\qed
\end{proof}

\noindent Theorem~\ref{lem:complexitymonproblem} is now an immediate corollary to Lemmas \ref{lem:complexity} and \ref{lem:monishard}.

\subsection{Proof of Lem.~\ref{lem:localmonsucc3}}
\label{app:lem:localmonsucc3}
We prove Lemma~\ref{lem:localmonincr}, which is Lemma~\ref{lem:localmonsucc3} restricted to two successors. We then sketch the proof for the case with three successors ($n=3$). Lemma~\ref{lem:localmonsucc3} can be proven for any other $n > 3$ in a similar manner.
\begin{restatable}{lemma}{lemlocalmonincr}
	\label{lem:localmonincr}
	Let $s \in S$ with $\Succ(s) = \{ s_1, s_2 \}$ and $s_2 \reachOrderEq s_1$. Furthermore let $\probdtmc(s,s_1)= f$ and $\probdtmc(s,s_2) = 1-f$ for $f\in\Q[V]$. Then for any parameter $p$:
	\[\monIncrLocal[p]{\sol[s]{T}{\pdtmc}} \quad\text{ iff }\quad \monIncr[p]{f}.\]
\end{restatable}

\begin{proof}[Lem.~\ref{lem:localmonincr}]
	\begin{align*}
	\monIncrLocal[p]{\sol[s]{T}{\pdtmc}} & 
	\overset{\mathtext{Def.~\ref{def:localMonIncr}}}{\iff}
	\left(\sum_{s' \in \Succ(s)} \sol[s']{T}{\pdtmc} \cdot \derivative{p}{\probdtmc(s,s')} \right)(\vec{u}) \geq 0 & \forall \vec{u} \in R \\
	& \iff \left(\sol[s_1]{T}{\pdtmc} \cdot \derivative{p}{f} + \sol[s_2]{T}{\pdtmc} \cdot \derivative{p}{(1-f)}\right)(\vec{u}) \geq 0& \forall \vec{u} \in R \\
	& \iff \derivative{p}{f}(\vec{u}) \cdot \underbrace{\left(\sol[s_1]{T}{\pdtmc} - \sol[s_2]{T}{\pdtmc}\right)(\vec{u})}_{\geq 0 \text{ as } s_2 \reachOrderEq s_1} \geq 0 & \forall \vec{u} \in R\\
	& \overset{}{\iff} \derivative{p}{f(\vec{u})} \geq 0& \forall \vec{u} \in R.
	\end{align*}
	\qed
\end{proof}

\lemLocalMonSucc*

For the case with $n = 3$ we have:
\begin{align*}
&\sol[s_3]{T}{\pdtmc} \leq \sol[s_2]{T}{\pdtmc} \leq \sol[s_1]{T}{\pdtmc} \text{ as } s_3 \reachOrderEq s_2 \reachOrderEq s_1, \text{ and}\\
&(f_1 + f_2 + f_3)(\vec{u})= 1 \text{ as the valuation should be well-defined.}
\end{align*}

We have to prove that:
\begin{align*}
\monIncrLocal[p]{\sol[s]{T}{\pdtmc}} &
\overset{\mathtext{Lem.~\ref{lem:localmonsucc3}}}{\iff} \begin{cases}
\monIncr{f_1} \text{, } \monDecr{f_2} \text{ and } \monDecr{f_3} & (i=1),\text{ or} \\
\monIncr{f_1} \text{, } \monIncr{f_2} \text{ and } \monDecr{f_3} &(i=2),\text{ or} \\
\monIncr{f_1} \text{, } \monIncr{f_2} \text{ and } \monIncr{f_3} &(i=3).
\end{cases}
\end{align*}

We have to show (Def.~\ref{def:localMonIncr}): 
\begin{align*}
 \left(\sum_{s' \in \Succ(s)} \sol[s']{T}{\pdtmc} \cdot \derivative{p}{\probdtmc(s,s')} \right)(\vec{u})\geq 0 & &\forall \vec{u} \in R
\end{align*}

For $i=1$ we make the following observation:
\begin{align*}
	& \left(\sum_{s' \in \Succ(s)} \sol[s']{T}{\pdtmc} \cdot \derivative{p}{\probdtmc(s,s')} \right)(\vec{u}) \\
	&  = \left(\sol[s_1]{T}{\pdtmc} \cdot \derivative{p}{f_1} + \sol[s_2]{T}{\pdtmc} \cdot \derivative{p}{f_2} + \sol[s_3]{T}{\pdtmc} \cdot \derivative{p}{f_3}\right)(\vec{u}) \\
	& =  \left(\sol[s_1]{T}{\pdtmc} \cdot \derivative{p}{(1-f_2-f_3)} + \sol[s_2]{T}{\pdtmc} \cdot \derivative{p}{f_2} + \sol[s_3]{T}{\pdtmc} \cdot \derivative{p}{f_3}\right)(\vec{u}) \\
	& =  \derivative{p}{f_2(\vec{u})}\cdot\left(\sol[s_2]{T}{\pdtmc} - \sol[s_1]{T}{\pdtmc}\right)(\vec{u}) + \derivative{p}{f_3(\vec{u})} \cdot \left(\sol[s_3]{T}{\pdtmc} - \sol[s_1]{T}{\pdtmc}\right)(\vec{u})
\end{align*}

We observe that: 
\[ \underbrace{\derivative{p}{f_2(\vec{u})}}_{\leq 0}\cdot \underbrace{\left(\sol[s_2]{T}{\pdtmc} - \sol[s_1]{T}{\pdtmc}\right)(\vec{u})}_{\leq 0} + \underbrace{\derivative{p}{f_3(\vec{u})}}_{\leq 0}\cdot \underbrace{\left(\sol[s_3]{T}{\pdtmc} - \sol[s_1]{T}{\pdtmc}\right)(\vec{u})}_{\leq 0} \geq 0\]

From this, it follows that Lem.~\ref{lem:localmonsucc3} holds if the condition for $i=1$ holds. The second case is proved in a similar manner.

For $i=3$ we observe the following:
\begin{align*}
\derivative{p}{(f_1 + f_2 + f_3)(\vec{u})} &= \derivative{p}{1}= 0
\end{align*}
Therefore, not all derivatives can be negative. 


For $n > 3$ the proof follows in a similar way.

\clearpage
\subsection{Proof of Theorem~\ref{thm:monIncr}}
\thmMonIncr*
In order to proof Theorem~\ref{thm:monIncr}, we first introduce the notion of paths in pMCs, and provide two auxiliary Lemma's on paths. Then we lift local monotonicity (Def.~\ref{def:localMonIncr}) to local monotonicity for $n$ steps.

An \emph{infinite path} of a pMC $\pdtmc$ is a non-empty infinite sequence $\pi = s_0 s_1 s_2 \ldots $ of states $s_i \in S$ such that $\probdtmc(s_i, s_{i+1}) > 0$ for $i\geq0$. 
A \emph{finite path} of a pMC $\pdtmc$ is a finite, non-empty prefix of an infinite path in $\pdtmc$.
We define the \emph{length of a finite path} $|\pi| = n$ for $\pi = s_0 s_1 \ldots s_n$, and let $\pi_i$ be the $i^{th}$ state of a path.
Let $\PathsLength{n}(s)$ ($\Paths(s)$) be the set of all finite paths with length $n$ (infinite paths) starting from state $s\in S$. Let $\Pr(\pi) = \sum_{i=0}^{n-1} \probdtmc(s_i, s_{i+1})$ be the probability of a finite path $\pi = s_0 s_1 \ldots s_n$.
This can be lifted to the probability of an infinite path via a cylinder set construction~\cite{BK08}.

For a finite path $\pi = s_0 s_1 \ldots s_n$ we let $ [\pi \models \finally T]$ be $1$ if there exists a state $s_i$ such that $s_i \in T$, and $0$ otherwise. 

From this we observe that for any $s\in S$:
\[\sol[s]{T}{\pdtmc}(\vec{u}) = \sum_{\pi \in \Paths(s)} \Pr(\pi) \cdot [\pi \models \finally T].\]

Let $Z \subseteq S$ be the set of zero states ($\forall s \in Z.\sol[s]{T}{} =0$). 
Observe that if $ [\pi \models \finally T] = 1$ ($ [\pi \models \finally Z]  = 1$) then $ [\pi \models \finally Z] = 0$ ($[\pi \models \finally T] = 0$). 
We can write the set of paths from $s \in S$ as the union of three disjoint sets. 
	
\begin{lemma}
		Let $\PathsLength{n}(s)$ = $\PathsLength{n}_T(s)\cup \PathsLength{n}_Z(s) \cup \PathsLength{n}_?(s)$ where $\PathsLength{n}_T(s) $, $\PathsLength{n}_Z(s) $, and $\PathsLength{n}_?(s)$  are disjoint, and
		\begin{itemize}
			\item $\PathsLength{n}_T(s) = \{\PathsLength{n}(s)\,|\, \pi \models \finally T\}$ , 
			\item $\PathsLength{n}_Z(s) = \{\PathsLength{n}(s)\,|\, \pi \models \finally Z\}$ , and
			\item $\PathsLength{n}_?(s) = \{\PathsLength{n}(s)\,|\, \pi \not\models \finally T \wedge \pi \not \models \finally Z\}$.
		\end{itemize}
		 
\end{lemma}
By definition of $T$ ($Z$), we obtain for $\PathsLength{n}_T(s)$ ($\PathsLength{n}_Z(s)$): $\sol[\pi_n]{T}{} = 1$ ($\sol[\pi_n]{T}{} = 0$). 
In a similar way, we can split the infinite paths. 
Observe that $\Paths_?(s) = \emptyset$.

\begin{lemma}
		Let $\Paths(s)$ = $\Paths_T(s)\cup \Paths_Z(s) \cup \Paths_?(s)$ where $\Paths_T(s) $, $\Paths_Z(s) $, and $\Paths_?(s)$  are disjoint, and
		\begin{itemize}
			\item $\Paths_T(s) = \{\Paths(s)\,|\, \pi \models \finally T\}$,
			\item $\Paths_Z(s) = \{\Paths(s)\,|\, \pi \models \finally Z\}$, and
			\item $\Paths_?(s) = \{\Paths(s)\,|\, \pi \not\models \finally T \wedge \pi \not \models \finally Z\} = \emptyset$.
		\end{itemize}
\end{lemma}

\begin{definition}[Locally monotonic increasing for n steps]
	\label{def:localMonIncrSteps}
	$\sol[s]{T}{\pdtmc}$ is \emph{locally monotonic increasing} for $n$ steps in parameter $p$ (at $s$) on region $R$, denoted $\monIncrLocalSteps[p]{\sol[s]{T}{\pdtmc}}{n}$, if for all $\vec{u} \in R$:
	\[\left(\sum_{\pi \in \PathsLength{n}(s)} \left(\derivative{p}{\Pr(\pi)}\right) \cdot \sol[\pi_n]{T}{\pdtmc}  \right)(\vec{u}) \geq 0.\]
	\emph{Locally monotone decreasing} for $n$ steps, denoted $\monDecrLocalSteps[p]{\sol[s]{T}{\pdtmc}}{n}$, is defined analogously.
\end{definition}

We observe that for $n=1$, Def.~\ref{def:localMonIncrSteps} corresponds to Def.~\ref{def:localMonIncr}. 

\begin{proof}[Theorem~\ref{thm:monIncr}]
	We showcase the proof for $\forall s \in S. |\Succ(s)| \leq 2$. The general case can be shown analogously.
	
	We want to show:
		\[\Big(\forall s \in S.\,\monIncrLocal[p]{\sol[s]{T}{\pdtmc}}\Big) \implies\monIncr[p]{\sol[s]{T}{\pdtmc}}.\]
		
	To that end, we make the following claims:
	\begin{align}
		\label{eq:localMonToNSteps}
		\Big(\forall s \in S.\,\monIncrLocal[p]{\sol[s]{T}{\pdtmc}}\Big) \implies \lim_{n\rightarrow\infty} \left( \monIncrLocalSteps[p]{\sol[s]{T}{\pdtmc}}{n}\right)
		\end{align}
and
	\begin{align}
		\label{eq:NStepsToGlobalMon}
		 \lim_{n\rightarrow\infty}\left(\left(\sum_{\pi \in \PathsLength{n}(\sinit)} \left(\derivative{p}{\Pr(\pi)}\right) \cdot \sol[\pi_n]{T}{\pdtmc}\right)(\vec{u}) \right) \geq 0\nonumber \\
		\implies \derivative{p}{\left(\sum_{\pi \in \Paths(\sinit)} \Pr(\pi) \cdot [\pi \models \finally T]\right)}(\vec{u}) \geq 0
	\end{align}
	Then we derive:
	\begin{multline*}
	\Big(\forall s \in S.\,\monIncrLocal[p]{\sol[s]{T}{\pdtmc}}\Big) \overset{\mathtext{Eq.~(\ref{eq:localMonToNSteps})}}{\implies}  \lim_{n\rightarrow\infty} \left( \monIncrLocalSteps[p]{\sol[s]{T}{\pdtmc}}{n}\right)\\
	\overset{\text{Def.~\ref{def:localMonIncrSteps}}}{\implies} \lim_{n\rightarrow\infty}\left(\left(\sum_{\pi \in \PathsLength{n}(\sinit)} \left(\derivative{p}{\Pr(\pi)}\right) \cdot \sol[\pi_n]{T}{\pdtmc}\right)(\vec{u}) \right) \geq 0\\
	\overset{\mathtext{Eq.~(\ref{eq:NStepsToGlobalMon})}}{\implies} \derivative{p}{\left(\sum_{\pi \in \Paths(\sinit)} \Pr(\pi) \cdot [\pi \models \finally T]\right)}(\vec{u}) \geq 0
	\overset{\mathtext{Def.~\ref{def:monotoneFunction}}}{\iff}\monIncr[p]{\sol[s]{T}{\pdtmc}}
	\end{multline*}
	It remains to show our claims.
	
	\paragraph{Proof of~(\ref{eq:localMonToNSteps}).}
	By induction on $n$ we show:
	\begin{equation*}
	\Big(\forall s \in S.\,\monIncrLocal[p]{\sol[s]{T}{\pdtmc}}\Big) \implies \forall n. \monIncrLocalSteps[p]{\sol[s]{T}{\pdtmc}}{n}.
	\end{equation*}
	Furthermore, we observe that:
	\[\left( \forall n. \monIncrLocalSteps[p]{\sol[s]{T}{\pdtmc}}{n} \right)\implies \lim_{n\rightarrow\infty} \monIncrLocalSteps[p]{\sol[s]{T}{\pdtmc}}{n}\]
	\emph{Base case: }$n=1$. 
	This immediately follows from Def.~\ref{def:localMonIncr}.\\
	\emph{Induction step:}
	Assume that for $n\leq k$, the following holds:
	\begin{equation}
	\Big(\forall s \in S.\,\monIncrLocal[p]{\sol[s]{T}{\pdtmc}}\Big) \implies \monIncrLocalSteps[p]{\sol[\sinit]{T}{\pdtmc}}{n}.\label{prf:lem:monIncr:IH}\tag{IH}
	\end{equation}	
	We want to show for $n=k+1$:
	\begin{equation*}
	\Big(\forall s \in S.\,\monIncrLocal[p]{\sol[s]{T}{\pdtmc}}\Big) \implies \monIncrLocalSteps[p]{\sol[s]{T}{\pdtmc}}{k+1}.
	\end{equation*}	
		Let $\pi = \pi'\pi_{k+1}$. From Def.~\ref{def:localMonIncrSteps} we obtain:
	\[ \monIncrLocalSteps[p]{\sol[s]{T}{\pdtmc}}{k+1} \iff \forall \vec{u} \in R.\,\left(\sum_{\pi \in \PathsLength{k+1}(s)} \left(\derivative{p}{\Pr(\pi)}\right) \cdot \sol[\pi_{k+1}]{T}{\pdtmc}  \right)(\vec{u}) \geq 0.\]

	As we are only considering graph-preserving valuations (and assume at most 2 successors), we distinguish two cases. 
	If $\pi_k$ has one successor, then $\probdtmc(\pi_k, \pi_{k+1}) = 1$. 
	Therefore, $\sol[\pi_{k}]{T}{\pdtmc} = \sol[\pi_{k+1}]{T}{\pdtmc}$. 
	For this the induction step holds. 
	If $\pi_k$ has two successors, then $\pi_{k+1} \in \{s_1, s_2\}$. Let $\probdtmc(\pi_k, s_1) = f$, and $\probdtmc(\pi_k, s_2) = 1-f$. 
	As $\monIncrLocal[p]{\sol[\pi_k]{T}{\pdtmc}}$, we obtain $s_2 \reachOrderEq s_1$ and $\monIncr[p]{f}$.

	\begin{align*}
	\allowdisplaybreaks
		&\sum_{\pi \in \PathsLength{k+1}}\left(\derivative{p}{\Pr(\pi)}\right) \cdot \sol[\pi_{k+1}]{T}{\pdtmc} \\
		&= \sum_{\pi' \in \PathsLength{k}}\left(\derivative{p}{\left(f\cdot \Pr(\pi')\right)}\right) \cdot \sol[s_1]{T}{\pdtmc} +\left(\derivative{p}{\left((1-f)\cdot \Pr(\pi')\right)}\right)\cdot \sol[s_2]{T}{\pdtmc}\\
		&\overset{\mathtext{Chain Rule}}{=} \sum_{\pi' \in \PathsLength{k}}\left(\left(\derivative{p}{f}\right)\cdot \Pr(\pi') + \left(\derivative{p}{\Pr(\pi')}\right) \cdot f \right)  \cdot \sol[s_1]{T}{\pdtmc} \\
		&\qquad+ 
		\left(\left(\derivative{p}{\left(1-f\right)}\right)\cdot \Pr(\pi') + \left(\derivative{p}{\Pr(\pi')}\right) \cdot (1-f)\right)\cdot \sol[s_2]{T}{\pdtmc}\\
		&= \sum_{\pi' \in \PathsLength{k}}
		\Pr(\pi') \cdot \left( \left(\derivative{p}{f}\right) \cdot \sol[s_1]{T}{\pdtmc} +\left(\derivative{p}{(1-f)}\right) \cdot \sol[s_2]{T}{\pdtmc}\right)\\
		&\qquad + \left(\left(\derivative{p}{\Pr(\pi')}\right)\cdot f \cdot \sol[s_1]{T}{\pdtmc} \right) + \left(\left(\derivative{p}{\Pr(\pi')}\right)\cdot (1-f) \cdot \sol[s_2]{T}{\pdtmc} \right)\\
		&= \sum_{\pi' \in \PathsLength{k}}
		\Pr(\pi') \cdot \left( \left(\derivative{p}{f}\right) \cdot \sol[s_1]{T}{\pdtmc} +\left(\derivative{p}{(1-f)}\right) \cdot \sol[s_2]{T}{\pdtmc}\right)\\
		&\qquad + \underbrace{\left(\left(\derivative{p}{\Pr(\pi')}\right)\cdot \sol[\pi_k]{T}{\pdtmc} \right)}_{\left( {f} \cdot \sol[s_1]{T}{\pdtmc} +\left({1-f}\right) \cdot \sol[s_2]{T}{\pdtmc}\right) = \sol[\pi_k]{T}{\pdtmc} }\\
		&=\underbrace{\sum_{\pi' \in \PathsLength{k}}
		\Pr(\pi') \cdot \left( \left(\derivative{p}{f}\right) \cdot \sol[s_1]{T}{\pdtmc} +\left(\derivative{p}{(1-f)}\right) \cdot \sol[s_2]{T}{\pdtmc}\right)}_{\geq 0\text{ as } s_2\reachOrderEq s_1 \text{ and } \monIncrLocal[p]{f}}\\
		&\qquad + \underbrace{\sum_{\pi' \in \PathsLength{k}}\left(\left(\derivative{p}{\Pr(\pi')}\right)\cdot \sol[\pi_k]{T}{\pdtmc} \right)}_{\geq 0 \text{ follows from~(\ref{prf:lem:monIncr:IH}})}
	\end{align*}

	\paragraph{Proof of (\ref{eq:NStepsToGlobalMon}).} We show the stronger statement:
	\begin{multline}
	\lim_{n\rightarrow\infty}\Bigg(\sum_{\pi \in \PathsLength{n}} \left(\derivative{p}{\Pr(\pi)}\right) \cdot \sol[\pi_n]{T}{\pdtmc}\Bigg) \\
 	= \derivative{p}{\left(\sum_{\pi \in \Paths} \Pr(\pi) \cdot [\pi \models \finally T]\right)}
 	\end{multline}
	We observe that: \[\lim_{n\rightarrow\infty}
	\left(\sum_{\pi \in \PathsLength{n}_?} \left(\derivative{p}{\Pr(\pi)}\right) \cdot \sol[\pi_n]{T}{\pdtmc}
	\right) = 0.\] We outline the steps below:
	\begin{align*}
	&\lim_{n\rightarrow\infty}\left(\sum_{\pi \in \PathsLength{n}} \left(\derivative{p}{\Pr(\pi)}\right) \cdot \sol[\pi_n]{T}{\pdtmc} \right)  \\
	&= \lim_{n\rightarrow\infty}
	\Big(\sum_{\pi \in \PathsLength{n}_T} \left(\derivative{p}{\Pr(\pi)}\right) \cdot \sol[\pi_n]{T}{\pdtmc} \\&\qquad\qquad+
	\sum_{\pi \in \PathsLength{n}_?} \left(\derivative{p}{\Pr(\pi)}\right) \cdot \sol[\pi_n]{T}{\pdtmc}
	\Big)  \\
	&= \lim_{n\rightarrow\infty}\left(\sum_{\pi \in \PathsLength{n}_T} \left(\derivative{p}{\Pr(\pi)}\right) \cdot \sol[\pi_n]{T}{\pdtmc}\right)   \\
	&= \sum_{\pi \in \Paths_T} \left(\derivative{p}{\Pr(\pi)}\right) \cdot [\pi \models \finally T]\\
	&= {\derivative{p}{\left(\sum_{\pi \in \Paths_T} \Pr(\pi) \cdot [\pi \models \finally T]\right)}} &{\text{as }[\pi \models \finally T] = 1} \\
	&= \derivative{p}{\left(\sum_{\pi \in \Paths} \Pr(\pi) \cdot [\pi \models \finally T]\right)}
	\end{align*}
	\qed
\end{proof}

\subsection{Proof of Lem.~\ref{lem:reachrelativesucc}}
\lemReachrelativesucc*
\begin{proof}[Lem.~\ref{lem:reachrelativesucc}]
	To prove Lem.~\ref{lem:reachrelativesucc} it is sufficient to show:
	\begin{multline*}
	\exists s' \in \Succ(s). \Succ(s) \subseteq [s'] \text{ and } s \in [s'] \\
	\iff \neg \left(\LB(\Succ(s)) \reachOrder[R,T] s \reachOrder[R,T] \UB(\Succ(s))\right)
	\end{multline*} 
	We make the following claim:
		\begin{equation}
			\label{eq:lem:reachrelativesucc:claim}
	s \in \LB(\Succ(s)) \vee s \in \UB(\Succ(s))
		\iff s \in \LB(\Succ(s)) \wedge s \in \UB(\Succ(s))
	\end{equation}
	Then we derive:
	\begin{align*}
	&\exists s' \in \Succ(s). \Succ(s) \subseteq [s'] \text{ and } s \in [s'] \\
	&\iff \forall s_1, s_2 \in \Succ(s).\sol[s_1]{T}{\pdtmc} = \sol[s_2]{T}{\pdtmc} \wedge \sol[s]{T}{\pdtmc} = \sol[s_1]{T}{\pdtmc}\\	
	&\iff \forall s' \in \Succ(s). s \equiv s'\\
	&\iff \forall s' \in \Succ(s). s \reachOrderEq s' \wedge s' \reachOrderEq s\\
	&\iff \underbrace{s \in \LB(\Succ(s)) \wedge s \in \UB(\Succ(s))}_{\UB(X) = \{ s \in S \mid X \reachOrderEq s \} \text{ and }\LB(X) = \{ s \in  S \mid s \reachOrderEq X \}} \\
	&\overset{\mathtext{Eq.~(\ref{eq:lem:reachrelativesucc:claim})}}{\iff} s \in \LB(\Succ(s)) \vee s \in \UB(\Succ(s))\\
	&\iff \neg \big(\LB(\Succ(s)) \reachOrder[R,T] s\big) \vee \neg\big( s \reachOrder[R,T] \UB(\Succ(s))\big)\\ 
	&\iff \neg \big(\LB(\Succ(s)) \reachOrder[R,T] s \reachOrder[R,T] \UB(\Succ(s))\big)
	\end{align*}
	
	We are left to prove our claim (\ref{eq:lem:reachrelativesucc:claim}):
	\begin{align*}
	& s \in \LB(\Succ(s)) \vee s \in \UB(\Succ(s))\\
	&\iff \forall s' \in \Succ(s) \forall \vec{u} \in R . \sol[s]{T}{\pdtmc}(\vec{u}) \leq \sol[s']{T}{\pdtmc}(\vec{u})\\
	&\qquad\qquad \vee \forall s' \in \Succ(s). \forall \vec{u} \in R .\sol[s]{T}{\pdtmc}(\vec{u}) \geq \sol[s']{T}{\pdtmc}(\vec{u})\\
	&\iff \underbrace{{\forall s' \in \Succ(s). \forall \vec{u} \in R .\sol[s]{T}{\pdtmc}(\vec{u}) = \sol[s']{T}{\pdtmc}(\vec{u})}}_{\text{as }\forall \vec{u} \in R. \sol[s]{T}{}(\vec{u}) = \sum_{s'\in\Succ(s)} \sol[s']{T}{}(\vec{u})} \\
	&\iff \forall s' \in \Succ(s). \forall \vec{u} \in R . \sol[s]{T}{\pdtmc}(\vec{u}) \leq \sol[s']{T}{\pdtmc}(\vec{u})	\\
	&\qquad\qquad\wedge \forall s' \in \Succ(s). \forall \vec{u} \in R .\sol[s]{T}{\pdtmc}(\vec{u}) \geq \sol[s']{T}{\pdtmc}(\vec{u}) \\
	&\iff s \in \LB(\Succ(s)) \wedge s \in \UB(\Succ(s)).
	\end{align*}
	\qed
\end{proof}

\subsection{Proof of Lem.~\ref{lem:forwardreasoning}}
\lemForwardreasoning*
\begin{proof}[Lem.~\ref{lem:forwardreasoning}]
	First of all, we observe that $\forall \vec{u} \in R$:
	\begin{align*}
	\sol[s]{T}{}(\vec{u})  &=\left( \probdtmc(s,s_1) \cdot \sol[s_1]{T}{}\right)(\vec{u})  +\left(\probdtmc(s,s_2) \cdot \sol[s_2]{T}{}\right)(\vec{u}) \\
	&=\left(\probdtmc(s,s_1) \cdot \sol[s_1]{T}{}\right)(\vec{u})  + \left((1-\probdtmc(s,s_1)) \cdot \sol[s_2]{T}{}\right)(\vec{u}) 
	\end{align*}
	\begin{enumerate}
		\item $s_1 \equiv s$. That is, $\forall \vec{u} \in R.\,\sol[s]{T}{}(\vec{u})  =\sol[s_1]{T}{}(\vec{u}) $. We obtain for each $\vec{u} \in R$: 
		\[
		\sol[s]{T}{}(\vec{u})  = \left(\probdtmc(s,s_1) \cdot \sol[s]{T}{}\right)(\vec{u})  + \left((1-\probdtmc(s,s_1)) \cdot \sol[s_2]{T}{}\right)(\vec{u})\]
		From which follows:
		\[\left((1-\probdtmc(s,s_1))\cdot \sol[s]{T}{}\right)(\vec{u})  = \left((1-\probdtmc(s,s_1)) \cdot \sol[s_2]{T}{}\right)(\vec{u}) \]

		So, \(\forall \vec{u} \in R.\,\sol[s]{T}{}(\vec{u})  =\sol[s_2]{T}{}(\vec{u})  \). Therefore, $s\equiv s_2$.
		\item $s_1\reachOrder s$. That is, $\forall \vec{u} \in R.\,\sol[s_1]{T}{}(\vec{u})   <\sol[s]{T}{}(\vec{u}) $. We obtain for each$ \vec{u} \in R$:
		\[\sol[s]{T}{}(\vec{u})  < \left(\probdtmc(s,s_1) \cdot \sol[s]{T}{}\right)(\vec{u})  + \left((1-\probdtmc(s,s_1)) \cdot \sol[s_2]{T}{}\right)(\vec{u}) \]
		From which follows:
		\[\left((1-\probdtmc(s,s_1))\cdot \sol[s]{T}{}\right)(\vec{u})  < \left((1-\probdtmc(s,s_1)) \cdot \sol[s_2]{T}{}\right)(\vec{u}) \]
		So, \(\forall \vec{u} \in R.\,\sol[s]{T}{}(\vec{u})  <\sol[s_2]{T}{}(\vec{u})  \). Therefore, $s\reachOrder s_2$.
		\item $s\reachOrder s_1$. That is, $\forall \vec{u} \in R.\,\sol[s]{T}{}(\vec{u})   <\sol[s_1]{T}{}(\vec{u}) $. We obtain for each $\vec{u} \in R$:
		\[\sol[s]{T}{}(\vec{u})  > \left(\probdtmc(s,s_1) \cdot \sol[s]{T}{}\right)(\vec{u})  + \left((1-\probdtmc(s,s_1)) \cdot \sol[s_2]{T}{}\right)(\vec{u}) \]
		From which follows:
		\[\left((1-\probdtmc(s,s_1))\cdot \sol[s]{T}{}\right)(\vec{u})  > \left((1-\probdtmc(s,s_1)) \cdot \sol[s_2]{T}{}\right)(\vec{u}) \]
		So, \(\forall \vec{u} \in R.\,\sol[s]{T}{}(\vec{u})  <\sol[s_2]{T}{}(\vec{u})  \). Therefore, $s\reachOrder s_2$.\qed
	\end{enumerate}
\end{proof}

\clearpage
\pagebreak

\section{Full Algorithm}
\label{app:completealg}

We consolidate the algorithm developed in Sect.~\ref{subsec:monousingro}--\ref{subsec:cycles}, resulting in Alg.~\ref{alg:full}.
To treat cycles, we can apply Lem.~\ref{lem:forwardreasoning} (l.~\ref{alg:line:beginSCCOrder}-\ref{alg:line:endSCCOrder}). 
Also, we may add a cycle breaking state $s'$ to the RO-graph (l.~\ref{alg:line:beginSCCOracle}-\ref{alg:line:endSCCOracle}). 
If assumptions are needed (l.~\ref{alg:line:elseAssumptions}-\ref{alg:line:continue}), then we push the three assumptions to the Queue. 
Therefore, l.~\ref{alg:line:push}, should not be executed, so we continue to the next iteration of the while loop after adding assumptions (l.~\ref{alg:line:continue})
At the end, we check for global monotonicity (l.~\ref{alg:line:beginMonCheck}-\ref{alg:line:endMonCheck}).

\begin{remark}
Any extension of an inconclusive ordering (recall Example~\ref{ex:inconclusiveorder}) is inconclusive too.
If the goal is to find a witness for global monotonicity, an orderings that is inconclusive for every parameter can be immediately discarded. 
\end{remark}

\ownsubsection{Discharging assumptions using a local NLP} Example~\ref{ex:NLP} showcases how a local NLP can be used to discharge assumptions.
\begin{example}
	\label{ex:NLP}
	Consider Fig.~\ref{fig:reqassum:pmc} and $s_2 \reachOrder s_3$.
	Let $s_4 \reachOrderAssumption s_5$. 
	For $R = (0.5,0.8) \times (0.1,0.3)$, the satisfiability of the following conjunction is checked, where $x_i$ encodes the reachability probability from state $s_i$:	
	\begin{align}
	0.5 & < p < 0.8 ~\land~ 0.1 < q < 0.3 \label{eq:exenc:region}\\
	x_2 &= p \cdot x_4 + (1{-}p) \cdot x_5 \label{eq:exenc:s1}\\
	x_3 &= q \cdot x_4 + (1{-}q) \cdot x_5 \label{eq:exenc:s2}\\
	0   & < x_4 < 1 ~\land~ 0 < x_5 < 1 \label{eq:exenc:probs}\\
	x_4 & < x_5 \label{eq:exenc:order}\\
	x_2 & \geq x_3 \label{eq:exenc:obligation}
	\end{align}
	Eq.~\eqref{eq:exenc:region} describes the region. Eqs.~\eqref{eq:exenc:s1}--\eqref{eq:exenc:s2} encode the reachability probabilities of $s_2$ and $s_3$, respectively. 
	For states $s_4, s_5$, we only know that the reachability probabilities are between $0$ and $1$ (Eq.~\eqref{eq:exenc:probs}), and $s_4 \reachOrder s_5$ (Eq.~\eqref{eq:exenc:order}).
	Finally, to validate the assumption $s_2 \reachOrder s_3$, we add constraint $\neg(s_2 \reachOrder s_3)$ (Eq.~\eqref{eq:exenc:obligation}).
	If the resulting constraint system has no solution, it follows $s_2 \reachOrder s_3$.
\end{example}

\ownsubsection{Treating cycles through SCC elimination.}
This method contracts each SCC into a set of states, one for each entry state of the SCC.
Applied to pMCs~\cite{DBLP:conf/qest/JansenCVWAKB14}, it preserves the reachability probabilities of target set $T$.
For each SCC, all non-entry states (i.e., states without incoming transitions from outside the SCC) are eliminated by state elimination~\cite{Daws04,param_sttt} and transitions between entry states are deleted~\cite{DBLP:conf/cav/DehnertJJCVBKA15}.
In the resulting acyclic pMC $\pdtmc'$, transitions from entry state $s$ of an eliminated SCC directly lead to states $t$ outside this SCC. 
Their transition probabilities in $\pdtmc'$ encode the (multi-step) reachability of reaching $t$ from $s$ in the cyclic $\pdtmc$.
This procedure works well if the pMC has several SCCs, but if the pMC is just a single SCC, this yields the complete solution function, which we intended to avoid.
If an SCC has many successor states (outside the SCC), this results in multiple successor states in $\pdtmc'$, possibly leading to multiple assumptions in Alg.~\ref{alg:latticeConstruction}+\ref{alg:assumptions:naive}.

\begin{algorithm}
	\caption{Monotonicity checking without assumption discharging}
	\label{alg:full}
	\begin{algorithmic}[1]
		\REQUIRE Acyclic pMC $\pDtmcInit$ 
		\ENSURE RO = a set of assumptions $\mathcal{A}$ with a monotonicity array
		\STATE Orders $\gets$ $\emptyset$
		\STATE Queue $\gets$ $\left(\mathcal{A} : \emptyset, \reachOrder : \{ (\bad,\good) \}, S' : S \setminus \{ \good, \bad \}\right)$
		\WHILE {Queue not empty}
		\STATE $\mathcal{A}, \reachOrderEqAssumption, S'$ $\gets$ Queue.pop()
		\IF{$S' = \emptyset$}
		\STATE Orders $\gets$ Orders $\cup$ $\{ (\mathcal{A}, \reachOrderEqAssumption) \}$.
		\ELSE
		\STATE select $s \in S'$ s.t. $s$ topologically last or if $s$ lies within an SCC and this SCC is topologically last 
		\IF{$\exists s' \in \Succ(s)$ s.t.\  $\Succ(s) \subseteq [s']$}
	
		\STATE extend RO-graph($\reachOrderEqAssumption$) with:  $s \equiv \Succ(s)$
		\ELSIF {$\reachOrderEqAssumption$ a total order for $\Succ(s)$}
		\label{alg:line:latticeConstruction:addBetween:new}
		\STATE extend RO-graph($\reachOrderEqAssumption$) with:\\ $s \reachOrderAssumption \min\UB(\Succ(s))$ and $\max\LB(\Succ(s)) \reachOrderAssumption s$ 
		\ELSIF {$\reachOrderEqAssumption$ is not a total order for $\Succ(s) = \{s_1, s_2\}$ and $s_1 \reachOrderAssumption s$ }\label{alg:line:beginSCCOrder}
		\STATE extend RO-graph($\reachOrderEqAssumption$) with: $s \reachOrderAssumption s_2$
		\ELSIF {$\reachOrderEqAssumption$ is not a total order for $\Succ(s) = \{s_1, s_2\}$ and $s_2 \reachOrderAssumption s$}
		\STATE extend RO-graph($\reachOrderEqAssumption$) with: $s_2 \reachOrderAssumption s$\label{alg:line:endSCCOrder}
		\ELSIF{$s$ part of an SCC and $\Succ(s) = \{s_1, s_2\}$}\label{alg:line:beginSCCOracle}
		\STATE pick a cycle breaking state $s'\in S'$ and $s'$ in SCC
		\STATE set $s$ to $s'$
		\STATE extend RO-graph($\reachOrderEqAssumption$) with: $\bad \reachOrderAssumption s$ and $s \reachOrderAssumption \good$\label{alg:line:endSCCOracle}
		\ELSE[$\reachOrderEqAssumption$ is not a total order for $\Succ(s)$]\label{alg:line:elseAssumptions}
		\STATE pick $t_1, t_2 \in \Succ(s)$ s.t.\ neither $t_1 \reachOrderEqAssumption t_2$ nor $t_2 \reachOrderEqAssumption t_1$
		\STATE Queue.push($(\mathcal{A}_\prec \cup \{ (t_1,t_2) \}, \mathcal{A}_\equiv), \reachOrderEqAssumption, S'$) 
		\STATE	Queue.push($(\mathcal{A}_\prec \cup \{ (t_2,t_1)  \}, \mathcal{A}_\equiv), \reachOrderEqAssumption, S'$) 
		\STATE Queue.push($(\mathcal{A}_\prec, \mathcal{A}_\equiv \cup \{ (t_1,t_2) \}), \reachOrderEqAssumption, S'$) 
		\STATE \textbf{continue} \label{alg:line:continue}
		
		\ENDIF
		\STATE Queue.push($\mathcal{A}, \reachOrderEqAssumption, S' \setminus \{ s \}$)\label{alg:line:push}
		\ENDIF

		\ENDWHILE
		\STATE
		\FOR {every order $\reachOrderEqAssumption$ in Orders} \label{alg:line:beginMonCheck}
		\STATE set for every parameter $p \in \Paramvar$ mon[$p$]$\rightarrow$\{T,T\}
		\FOR {every parametric state $s$}
		\STATE let $\Succ(s) = s_1, \ldots, s_n$, sorted based on $\reachOrderEqAssumption$
		\FOR {every parameter $p \in \Paramvar$ occuring at $s$, and mon[p] $\neq$ \{F,F\}}
		\IF {$\not\exists i \in [1, \hdots, n].\Big(\forall j \leq i.\; \monIncr[p]{f_j} \text{ and } \forall j > i.\; \monDecr[p]{f_j}\Big)$}
		\STATE mon[$p$][0] = F
		\ENDIF
		\IF{$\not\exists i \in [1, \hdots, n].\Big(\forall j \leq i.\; \monDecr[p]{f_j} \text{ and } \forall j > i.\; \monIncr[p]{f_j}\Big)$ }
		\STATE mon[$p$][1] = F\label{alg:line:endMonCheck}
		\ENDIF
		\ENDFOR
		\ENDFOR
		Result $\gets$ Result $\cup$ $\{ (\mathcal{A}, \text{mon}) \}$
		\ENDFOR
		
		\RETURN Result
	\end{algorithmic}
\end{algorithm}

\end{document}